%% file: main.tex
\newif\ifdraft \drafttrue
\documentclass{sigplanconf}

\input{prelude}

\usepackage[firstpage]{draftwatermark}
\SetWatermarkText{\hspace*{8in}\raisebox{3.5in}{\includegraphics[scale=0.1]{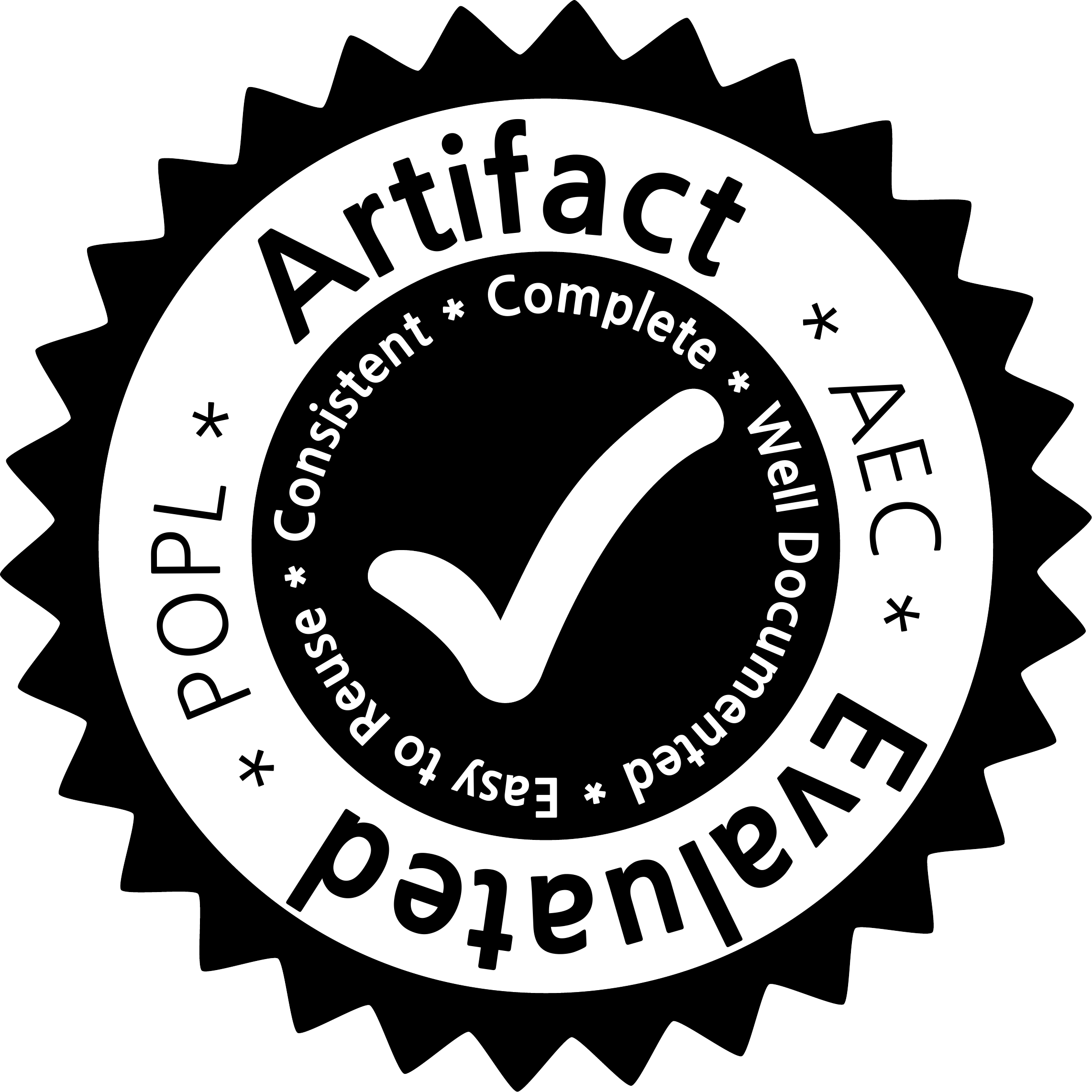}}}
\SetWatermarkAngle{0}

\begin{document}

\setlength{\pdfpageheight}{\paperheight}
\setlength{\pdfpagewidth}{\paperwidth}

\conferenceinfo{POPL~'15}{January 15--17, 2015, Mumbai, India}
\copyrightyear{2015} 
\copyrightdata{978-1-4503-3300-9/15/01} 
\doi{10.1145/2676726.2677000} 

\exclusivelicense

\titlebanner{DRAFT}        
\preprintfooter{DRAFT}     

\title{Higher-Order Approximate Relational Refinement Types \\
  for Mechanism Design and Differential Privacy}
\authorinfo
{
  Gilles Barthe$^{\star}$ \and
  Marco Gaboardi$^{\ddagger}$ \and
  Emilio Jes\'us Gallego Arias$^{\#}$ \\
  Justin Hsu$^{\#}$ \and
  Aaron Roth$^{\#}$ \and
  Pierre-Yves Strub$^{\star}$
}
{
  $^{\star}$ IMDEA Software Institute \and
  $^{\ddagger}$ University of Dundee, Scotland \and
  $^{\#}$ University of Pennsylvania
}
{}

\maketitle

\begin{abstract}
\emph{Mechanism design} is the study of algorithm design where the inputs to
the algorithm are controlled by strategic agents, who must be
\emph{incentivized} to faithfully report them. Unlike typical programmatic
properties, it is not sufficient for algorithms to merely satisfy the
property---incentive properties are only useful if the strategic agents
also \emph{believe} this fact.

Verification is an attractive way to convince agents that the incentive
properties actually hold, but mechanism design poses several unique challenges:
interesting properties can be sophisticated \emph{relational} properties of
probabilistic computations involving expected values, and mechanisms may rely on
other probabilistic properties, like \emph{differential privacy}, to achieve
their goals.

We introduce a relational refinement type system, called \THESYSTEM,
for verifying mechanism design and differential privacy. We show that
\THESYSTEM is sound w.r.t.\, a denotational semantics, and correctly
models $(\epsilon,\delta)$-differential privacy; moreover, we show
that it subsumes \DFuzz, an existing linear dependent type system for
differential privacy. Finally, we develop an SMT-based implementation
of \THESYSTEM and use it to verify challenging examples of mechanism
design, including auctions and aggregative games, and new proposed examples from
differential privacy.
\end{abstract}

\category{D.3.1}{Programming Languages}{Formal Definitions and Theory}[Semantics]
\category{D.2.4}{Software Engineering}{Software/Program Verification}.

\keywords{\hspace{-1.5ex} program logics; probabilistic programming}

\section{Introduction} \label{sec:intro}
When designing algorithms, we usually assume that the inputs are correctly
reported. However, in the real world, inputs may be provided by people who may
want to influence the outcome of the algorithm. \emph{Mechanism design} is the
field of algorithm design where the \emph{inputs} to the algorithm (often called
a \emph{mechanism}) are controlled by \emph{strategic agents} who may manipulate
what their inputs. In this setting, it is not enough to design an algorithm
which behaves correctly on correct input; the design of the mechanism must
convince (\emph{incentivize}) agents to provide their \emph{correct} inputs to
the algorithm.

The canonical application of mechanism design is auction design. In an auction,
the algorithmic problem can be very simple: for instance, allocate some set of
goods amongst a set of $n$ agents so as to maximize their sum value for the
goods. The inputs to the algorithmic problem are simply the agents' values for
the goods, but these are unknown to the algorithm designer. Instead, the
mechanism must elicit bids. To incentivize agents to bid honestly, actions
compute a price that each agent must pay. Auctions are generally designed so
that the allocation and payment rules incentivize agents to bid their true value
for the goods, \emph{no matter what their opponents do}. An auction that
satisfies this property is said to be \emph{dominant strategy incentive
  compatible}, or simply \emph{truthful}.  This is among the most important
\emph{solution concepts}---predictions about behavior of strategic agents---in
mechanism design.

Beyond auctions, mechanisms design can be used to handle more abstract games
where agents have a variety of actions and a real-valued \emph{utility function}
based on the actions selected by all agents. In most settings, agents do not
have any dominant strategies, and so we must be satisfied with weaker
\emph{solution concepts} like \emph{Nash equilibria}.  Informally, a set of
actions, one for each player, forms a Nash equilibrium if no player can increase
her utility unilaterally deviating to a different action \emph{so long as no
other player deviates}.\footnote{%
  Contrast this with dominant strategies: in a truthful auction, no agent can
  gain by deviating \emph{no matter what the other players play}.}
The hope is that strategic agents will collectively decide to play at
an equilibrium: no single agent can gain by deviating.

However, Nash equilibria can be an unrealistic prediction of behavior. First of
all, they are generally not unique: agents must somehow coordinate to play at a
single equilibrium, but different agents might prefer different equilibrium
outcomes.  Second of all, in games with a large number of players, agents
generally do not have complete knowledge about everyone's utility functions, and
so may not even know what the Nash equilibria of the game are. To help players
coordinate on an equilibrium, one approach is to design an equilibrium selection
mechanism. Agents are asked to report their utility functions to some
\emph{mediator}, and the mediator suggests some action for them to play. Agents
are strategic, so they are free to misreport their utility function, or
disregard the mediator's suggestion. A well-designed mediator will incentivize
agents to report truthfully and follow the recommendation.

A promising and recent tool for mediator design is \emph{differential
  privacy}~\citep{Dwork06}. The original goal of differential privacy was to
protect individuals privacy in data mining, by ensuring that answering the same
query on two databases differing in a single person's data leads to results that
are close in some sense. Seen another way, differential privacy limits any
individual's \emph{influence} on the result. This can be quite useful as a tool
for mechanism design: If the mediator satisfies differential privacy, agents
will have little incentive to deviate from truthful behavior since they can only
change the selected equilibria to a small degree.

While this is a theoretically clean idea, there can be practical issues:
currently proposed mechanisms are complex enough that agents may not be able to
verify the promised incentive properties of the mediator. In general, if agents
do not believe the incentive properties of a mechanism, then they may behave in
unpredictable ways or decline to participate.  Indeed, designers of an upcoming
public radio spectrum auction have stressed ``obviousness'' as a key
feature---the incentive features should be plainly apparent to any
agent~\citep{milgrom2014deferred}.  While this is a desirable goal for simple
mechanisms, it is hard to achieve for complex mechanisms. For these cases, we
propose an alternative approach: rather than simplify the mechanism, use formal
verification to automatically check incentive properties.

However, mechanism design poses serious challenges for verification.  First,
both differential privacy and equilibria properties are \emph{relational}
properties of programs, which reason about more than one run of the same
program; for instance, truthfulness states that the payoff of an agent in a run
when she reports truthfully is at least its payoff in a run where she reports
arbitrarily. Second, equilibria properties are significantly more involved than
typical program verification properties. For randomized mechanisms, properties
are stated in terms of \emph{expected value} rather than more standard
equivalences or relations between distributions. Finally, incentive properties
of mediator mechanisms rest on non-trivial interactions between game-theoretic
properties and differential privacy, so their formal verification must be
conducted in a framework that is expressive enough to reason about differential
privacy, game-theoretic properties, and interactions of the two. 

\paragraph*{Contributions}
To handle these challenges, we present \THESYSTEM, a type-based framework for
relational properties of higher-order probabilistic programs, like differential
privacy, truthfulness, and approximate equilibrium.  \THESYSTEM is based on
refinement types---an expressive type discipline that captures fine-grained
properties of computations by enriching types with
assertions~\citep{FreemanP91}---and tightly integrates several features:
relational refinements~\citep{rfstar}, refinements at higher types (where
assertions constrain the behavior of functions), and a polymonadic
representation of approximate refinements for probabilistic
computations~\citep{POPL:BKOZ12}. We demonstrate the theoretical and practical
relevance of \THESYSTEM through the following contributions:
\begin{itemize}
\item We demonstrate that \THESYSTEM achieves desirable meta-theore\-tical
  properties, like soundness with respect to a denotational semantics
  (\Cref{sec:semantics}) and semantic subtyping. As a contribution of
  independent interest, we show that the logical interpretation of
  non-termination adopted by existing refinement type systems is inconsistent
  with semantic subtyping when refinements at higher order types are allowed.

\item We define a type-preserving embedding of \DFuzz~\citep{GaboardiHHNP13}---a
  linear dependent type system for differential privacy---into \THESYSTEM, and
  recover soundness of \DFuzz from soundness of \THESYSTEM (\Cref{sec:dfuzz}).
  The embedding illustrates how semantic subtyping and refinements at higher
  types combine to internalize logical relations in \THESYSTEM.

\item We implement a type-checker for \THESYSTEM and verify examples drawn from
  differential privacy and mechanism design. For instance, we verify
  truthfulness properties of randomized auctions, and an equilibrium selection
  algorithm for aggregative games based on differential privacy.  The
  implementation is fully automated, discharging assertions using SMT solvers.
\end{itemize}
We discuss related work in \Cref{sec:relatedwork}, and conclude with possible
future directions in \Cref{sec:future}.

\section{Relational refinements, informally}\label{sec:overview}

We will establish properties of programs by using refinement types, an
expressive typing discipline introduced by~\citet{FreemanP91}. As is typical in
refinement type systems, we type expressions in two steps. First, we define a
simply typed system in which the type of probabilistic computations are modeled
using a probability monad $\tmod{\cdot}$; for example, the expected value of a
positive real-valued function w.r.t.\, a distribution is modeled by a function
$\Expect$ of type:
$\tmod{T} \rightarrow (T \rightarrow \rplus) \rightarrow \rplusinfty $,
where $\rplus$ denotes the type of positive reals and $\rplusinfty$ denotes
$\rplus$ extended with $+\infty$.

Next, we define a relational refinement type system for simply typed
expressions. Relational refinements~\citep{rfstar} specify properties
of \emph{pairs} of values via types of the form $\rtref{x}{T}{\phi}$,
where $\phi$ is a \emph{relational assertion}: a logical formula that
can express facts involving the \emph{left} instance $\l{x}$ and the
\emph{right} instance $\r{x}$ of $x$.  For instance, the type
$\rtref{x}{\mathbb{N}}{|\l{x} - \r{x} | \leq k}$ models pairs of
natural numbers which differ by at most $k$.

Both traditional and relational refinement type systems
(e.g.,~\citep{rfstar}) often forbid refinements at higher types, like
$\rtref{x}{T}{\phi}$ where $T$ is a function type. However, such
relational refinements are convenient to model properties of
probabilistic operators. For instance, given a distribution
$\mu :: \tmod{T}$, the following types for $\Expect\,\mu$ capture
monotonicity and linearity of expectation:
\[
\begin{array}{c}
\rtref{f}{T\rightarrow \rplus\!}{\forall z. \l{f}z\leq\r{f}z}
\rightarrow \rtref{r}{\rplusinfty}{\l{r}\leq\r{r}} \\[1mm]

\rtref{f}{\! T \rightarrow \rplus \!}{\!\forall
  z. \l{f}z\!=\! k\!\cdot\!\r{f}z}
\rightarrow \rtref{x}{\!\rplusinfty}{\!\l{x}\!=\!k\!\cdot\!\r{x}}
\end{array}
\]
Relational refinements can also be used to model relations between pairs of
distributions, like differential privacy.  A probabilistic computation
$F:T\rightarrow U$ is \emph{$(\epsilon, \delta)$-differentially private}
(w.r.t\, an adjacency relation $\Phi$) if for every $t_1,t_2\in T$, and for
every subset of outputs $E$,
\[
  t_1\mathrel{\Phi} t_2 \Longrightarrow \Pr_{x\leftarrow F~t_1}[x\in E] \leq
    \exp(\epsilon) \Pr_{x\leftarrow F~t_2}[x\in E] +\delta .
\]
The parameters $\epsilon$ and $\delta$ are non-negative real numbers
controlling the strength of the privacy guarantee. Using relational
refinement types, the type of $(\epsilon,\delta)$-differentially
private computations from $T$ to $U$ is:
\[
  \rtref{x}{T}{\l{x}\mathrel{\Phi}\r{x}}\rightarrow \rtref{\mu}{\tmod{U}}{
  \Delta_{\epsilon}(\l{\mu},\r{\mu}) \leq \delta} ,
\]
where%
\[
  \Delta_{\epsilon}(\mu_1,\mu_2) = \max_{E\subseteq U} \left(
    \Pr_{x\leftarrow\mu_1} [x\in E] - \exp(\epsilon) \Pr_{x\leftarrow\mu_2} [x
        \in E] \right)
\]
is the \emph{$\epsilon$-distance}~\cite{POPL:BKOZ12} between two distributions
$\mu_1$ and $\mu_2$ over $U$.

However, this modeling is not appropriate for practical program
verification; indeed, the definition of $\epsilon$-distance uses
probabilities and exponentials; it is possible to formalize basic
properties for these concepts, but more advanced reasoning, which is
required for some examples, is beyond the abilities of SMT solvers.
Instead, we introduce a probabilistic polymonad~\citep{polymonad},
with a type constructor of the form $\rtmod{\epsilon}{\delta}{\cdot}$
and two operators \textsf{unit} and \textsf{bind} with respective
types: 
$$\begin{array}{rcl} 
\textsf{unit} & : & T\rightarrow
\rtmod{\epsilon}{\delta}{T} \\ 
\textsf{bind} & : &
\rtmod{\epsilon}{\delta}{T} 
\rightarrow (T \rightarrow \rtmod{\epsilon'}{\delta'}{U}) 
\rightarrow \rtmod{\epsilon + \epsilon'}{\delta+\delta'}{U} .
\end{array}$$
The main advantage of the polymonad versus the explicit formalization
of $\epsilon$-distance is that all reasoning about probabilities is
confined to the definition of valid judgment, and to the proof of
soundness of the monadic rules. On the other hand, the refinement
types $T$ and $U$ remain standard relational refinements, and
do not need to refer to probabilities or exponentials.

The interpretation of $\rtmod{\epsilon}{\delta}{\cdot}$ is based on a lifting
operator $\lift{\epsilon,\delta}{\cdot}$ that turns a relation $\Psi$ on
$T_1\!\times\!T_2$ into a relation $\lift{\epsilon,\delta}{\Psi}$ on
$\tmod{T_1}\!\times\!\tmod{T_2}$; we will provide the formal definition in
\Cref{subsec:refinement}.

A useful property of lifting~\cite{POPL:BKOZ12} is that two distributions are
related by $\lift{\epsilon,\delta}{=}$ iff their $\epsilon$-distance is upper
bounded by $\delta$. In particular, $(\epsilon,\delta)$-differentially private
computations can be modeled by the relational refinement:
\[
  \rtref{x}{T}{\l{x}\mathrel{\Phi}\r{x}}\rightarrow
\rtmod{\epsilon}{\delta}{\rtref{y}{U}{\l{y} = \r{y}}} ,
\]
which leads to simpler verification conditions over individual outputs rather
than sets of outputs.

Another advantage of this polymonadic approach is that the type of the
\textsf{bind} operator captures the sequential composition theorem of
differential privacy~\citep{Dwork06}, and that it leads to an elegant type
system in which quantitative reasoning related to differential privacy is
confined to the rules of the polymonad. This is in contrast to some prior work
where quantitative reasoning is pervasive in all rules of the
system~\cite{POPL:BKOZ12}.

The last component of \THESYSTEM is a monad $\stmodc{\cdot}$ to model diverging
computations. While this is quite standard, this approach is key to reconcile
semantic subtyping, and refinements at higher types. We elaborate on this point
 next.

\subsection{An aside on non-termination and refinement types}
\label{app:inconsistency}
Exploiting the power of refinement types requires the ability to draw
useful inferences from assertions. These inferences are typically represented in
typing derivations via a subtyping relation $\preceq$, and applied with a special
typing rule, called \emph{subsumption}, that changes the type of an
expression to an arbitrary supertype. Ideally, the subtyping relation
should be complete w.r.t.\, the denotational semantics of the type
system---a property known as \emph{semantic subtyping}. However, extending
most existing refinement type systems with semantic subtyping can lead
to an inconsistency, because semantic subtyping conflicts with
typical logical treatments of non-termination.

Inconsistency can arise both in non-relational and relational
settings; let us consider the non-relational case.  Here,
inconsistency manifests itself as an expression of type
$\tref{x}{\stnat}{\bot}$ in the empty context, with the expression
reducing to a value. Existing refinement type systems such as \fseven,
\fstar and \liquidH  assign the type $\tref{x}{\stnat}{x\neq 0}
\rightarrow {\tref{y}{\stnat}{\bot}}$ to any recursive function $f$
that does not terminate on values $x\neq 0$, such as for instance
$ \sletrec{}{g}{x}{\scase{x}{[ 0 \Rightarrow 0 \mid s~y \Rightarrow g~x]}} $.
On the other hand, semantic subtyping validates the equivalence
\textsc{Fun-Sub}:
$$
\{x : T \!\vbar\! \phi\}\! \rightarrow \!\{y : U \!\vbar\! \psi\}
\simeq  \{f : T\!\rightarrow \!U \vbar \forall x:T.\phi
\!\Longrightarrow\!\psi[f~x/y]\}
$$
where we write $T\simeq U$ iff $T\preceq U$ and
$U\preceq T$. It follows that $f~0$ has the problematic type:\\

\AxiomC{$\vdash f: \tref{x}{\stnat}{x\neq 0} \rightarrow
\tref{y}{\stnat}{\bot}$} \UnaryInfC{$\vdash f:\tref{g}{\stnat \rightarrow
\stnat}{ \forall x:\stnat.x\neq 0 \Longrightarrow \bot}$} \UnaryInfC{$\vdash
f:\tref{g}{\stnat \rightarrow \stnat}{ \forall x:\stnat. \bot}$}
\UnaryInfC{$\vdash f: \stnat \rightarrow \tref{x}{\stnat}{\bot}$}
\AxiomC{$\vdash 0:\stnat$} \BinaryInfC{$\vdash f~0:\tref{x}{\stnat}{\bot} ,$}
\DisplayProof \\[2ex]
where the first three inferences are by subsumption (the first two by
\textsc{Fun-Sub}, and the second by the standard rule of consequence---replacing
an assertion by a logically weaker one).

This example shows that a naive combination of semantic subtyping with
refinements at higher types is inconsistent with the use of $\bot$ to model
diverging computation. We will follow a more semantically correct approach, by
modeling non-termination as a (monadic) effect using a monad $\stmodc{\cdot}$
that distinguishes non-terminating and terminating computations. In this way,
the inconsistency is avoided.

We conclude this discussion by noting that this counterexample is independent of
the evaluation strategy. Accordingly, it complements the recent observation by
\citet{Vazou+14:ICFP} that a logical interpretation of non-termination is
unsound for call-by-name, even in languages without higher-order refinements.
\citet{Vazou+14:ICFP} solve the issue by internalizing size types into
refinement types to enforce termination.

\section{The \THESYSTEM~System} \label{sec:system}
\subsection{Expressions}
\THESYSTEM{} is a relational type discipline for a $\lambda$-calculus
with inductive types, unbounded recursion and monads for probabilities
and partiality. For readability, we only present our calculus with
some fixed inductive types.

Let $\vars = \{ x, y, \ldots \}$ be a countably infinite set of
variables. The set $\pcf(\vars)$ of expressions with variables in $\vars$ is
defined as follows:
\begin{equation*}
  \begin{array}{@{}r@{\hspace{0.3em}}rl@{}}
    e
      & ::=   & x \vbar n \in \mathbb{N} \vbar \alpha \in \rplusinfty
                  \vbar \sone \vbar \snil \vbar \scons{\cdot}{\cdot}
                  \vbar \sfalse \vbar \strue\\
      & \vbar & e\ e \vbar \slam{x}{e} \vbar \slet{x}{e}{e} \vbar
                \sletrec{\markvar}{f}{x}{e}\\
      & \vbar & \sif{e}{e}{e} \vbar
                \sifseq{e}{x}{x}{e}{e}\\
      & \vbar & \sunitC{e} \vbar \sbindC{x}{e}{e} \vbar \sunitM{e} \vbar
      \sbindM{x}{e}{e} ,
  \end{array}
\end{equation*}
\noindent where $\markvar \in \{ \snmark, \cdot \}$ and $\rplusinfty$ stands
for $\rplus$ augmented with the infinity $\infty$. We write $\pcf$ for
$\pcf(\vars)$ when $\vars$ is clear from the context.

Most of the syntax is standard. The expressions $\sunitM{e}$ and
$\sbindM{x}{e}{e}$ corresponds to the unit and the multiplication of
the \emph{probabilistic monad}. Similarly, $\sunitC{e}$ and
$\sbindC{x}{e}{e}$ corresponds to the unit and the multiplication of
the \emph{partiality monad}. Finally, we have two expressions for
building recursive definitions: one for terminating programs, the
other for non terminating ones. We distinguish between them by means
of the superscript $\markvar \in \{ \snmark, \cdot \}$.

\THESYSTEM distinguishes between \emph{expressions} and
\emph{relational expressions}. The former are used in the subject of
typing judgments, and correspond to the actual programs to which we
can assign semantics. The latter are used in assertions.

\begin{definition}[Expressions and Relational Expressions]
Let $\rvars$ and $\pvars$ be two disjoint countably infinite sets of
\emph{relational} and \emph{plain} variables.
Associated with every relational variable $x\in \rvars$, we have a left instance
$\l{x}$ and a right instance $\r{x}$.
We write $\rmark{\rvars}$ for $\bigcup_{x \in \rvars} \{ \l{x}, \r{x} \}$
and $\rmark{\vars}$ for $\rmark{\rvars} \cup \pvars$.
The set of \THESYSTEM~\emph{expressions} $\expr$ is the set of expressions
in $\pcf(\pvars)$.
The set of \THESYSTEM~\emph{relational expressions} $\rexpr$ is the
set of expressions in $\pcf(\rmark{\vars})$, where only non-relational 
variables can be bound.
\end{definition}

\subsection{\THESYSTEM~Types}
We introduce the types of \THESYSTEM in two steps. First, we introduce
simple types; for simplicity, we restrict instances of inductive types
to base types. Then, we introduce relational refinement types, which
express properties about two interpretations of an expression.
\begin{definition}[Types]
The sets $\pcfty$ and $\pcfcoty$ of (simple) types and core (simple)
types are defined as follows:

\begin{center}
  $\begin{array}{rcl}
    \tau, \sigma, \ldots \in \pcfty
      & ::= & \coty{\tau} \vbar
              \stmod{\tau} \vbar
              \stmodc{\tau} \vbar
              \stfun{\tau}{\sigma}\\[.2em]
    \coty{\tau}, \coty{\sigma}, \ldots \in \pcfcoty
      & ::= & \stunit \vbar \stbool \vbar \stnat \vbar \stxreal \vbar
              \rplusinfty \vbar
              \stlist{\coty{\tau}} .
  \end{array}$
\end{center}
\end{definition}
The type $\stmod{\tau}$ corresponds to the probability monad over the
type $\tau$, while the type $\stmodc{\tau}$ corresponds to the
partiality monad over the type $\tau$. Besides the standard function
types $\stfun{\tau}{\sigma}$, the type language includes the unit
type, booleans, integers, reals and lists.

Relational types extend the grammar of simple types with relational
refinements, and use a dependent function type rather than standard function
types.
\begin{definition}
The sets of \emph{relational types} $\rtypes = \{ T, U, \ldots \}$ and
\emph{assertions} $\rassert = \{ \phi, \psi, \ldots \}$ are defined as
follows:

\begin{equation*}
\begin{array}{r@{\hspace{0.3em}}l}
    T, U \in \rtypes
      ::= & \coty{\tau} \vbar
              \rtmod{\epsilon}{\delta}{T} \vbar
              \rtmodc{T} \vbar
              \rtprod{x}{T}{T} \vbar
              \rtref{x}{T}{\phi}\\[.3em]
    \phi, \psi\in \rassert
      ::= &
              \begin{array}[t]{@{}l@{\hspace{1cm}}l@{}}
                \fquant{\quantvar}{x}{\tau}{\phi} & (x \in \pvars)\\[.3em]
                \rquant{\quantvar}{x}{T}{\phi} & (x \in \rvars)\\[.3em]
              \end{array}\\
           & \rformc(\phi_1, \ldots, \phi_n) \vbar
              \rmark{e} = \rmark{e} \vbar
              \rmark{e} \le \rmark{e}\\[.3em]
    \rformc = & \{ \sfrac{\rtrue}{0}, \sfrac{\rfalse}{0},
                     \sfrac{\neg}{1}, \sfrac{\vee}{2}, \sfrac{\wedge}{2},
                     \sfrac{\implies}{2} \} , 

  \end{array}
\end{equation*}
\noindent where $\epsilon, \delta, \rmark{e} \in \rexpr$,
and $\quantvar\in\{\forall, \exists\}$.
\end{definition}
The definitions of relational types and assertions are mutually
recursive. For the latter, the constructors $\rtmodc{.}$ and
$\rtmod{\epsilon}{\delta}{.}$ capture the partiality
monad and the probability polymonad for relational refinements.  The
type $\rtprod{x}{T}{U}$ corresponds to the dependent type product of
$T$ over $U$ indexed by $x$.  As usual, we write $\stfun{T}{U}$ for
$\rtprod{x}{T}{U}$ when $x$ does not occur free in $U$.  A type of the
shape $\rtref{x}{T}{\phi}$ \emph{refines} the type $T$ using the
assertion $\phi$.  In both dependent and refinement
types, we require the bound variable to be relational ($x\in\rvars$).

Assertions are built from primitive assertions using the standard connectives
and quantification; we allow quantification over both relational and plain
variables.  Primitive assertions are equalities and inequalities over relational
expressions.

The notions of substitutions are defined largely as usual.  Substitutions bind
\emph{pairs} of (non-relational) expressions to relational variables, and
involve a special treatment of \emph{refinement type constructors}, which must
handle the relational expressions.
For instance, for a substitution $\rho = \{ y \mapsto (e_1, e_2) \}$, we define

\begin{center}
  $\begin{array}{r@{\,}c@{\,}l}
     \rtref{x}{T}{\phi}\rho & = &
       \left\{ x:: T\rho \,\left\vert\, \phi\left\{
         \begin{array}{@{}l@{}}
           \l{y} \mapsto \l{\rembed{e_1}}\\
           \r{y} \mapsto \r{\rembed{e_2}}
         \end{array}
       \right\} \right. \right\} ,
   \end{array}$
\end{center}

\noindent where $\l{\rembed{e}}$ (resp. $\r{\rembed{e}}$) is obtained from
$e$ by replacing all the free variables $x$ of $e$ by $\l{x}$
(resp. $\r{x}$).

\subsection{Standard typing} \label{sec:semantics}
We define the simply typed layer of \THESYSTEM, and prove its
soundness w.r.t.\, a denotational semantics. Both the type
system and the semantics are mostly standard.

\paragraph{Static semantics}
The sole notable typing rules are for the two different $\kwletrec$~of our
language:
$$
  \inferrule*[left=LetRec]
    {\Gamma, f : \stfun{\tau}{\stmodc{\sigma}} \vdash \slam{x}{e} : \stfun{\tau}{\stmodc{\sigma}}}
    {\Gamma \vdash \sletrec{}{f}{x}{e} : \stfun{\tau}{\stmodc{\sigma}}}
$$
$$
   \inferrule*[left=LetRecSN]
    {\Gamma, f : \stfun{\tau}{\sigma} \vdash \slam{x}{e} : \stfun{\tau}{\sigma} \\
     \sncond}
    {\Gamma \vdash \ssnletrec{f}{x}{e} : \stfun{\tau}{\sigma} .}
$$
The rule \textsc{LetRec} handles unrestricted recursion, and requires
that the output type be in the partiality monad. On the contrary, the
rule \mbox{\textsc{LetRecSN}} does not impose any restriction on the
type, but the expression must pass a \emph{termination guard}. We leave the
termination guard unspecified; possible forms of enforcing the guard
include sized types and syntactic criteria.

\paragraph{Denotational semantics}
The denotational semantics is largely standard.  Core types are
interpreted in the standard way. The interpretation of types mixes a
set-theoretical and cpo semantics in order to accommodate the
partiality monad.
\begin{definition}[Interpretation of types] The \emph{interpretation}
$\tyinterp{\tau}$ of a type $\tau\in\pcfty$ is inductively defined as follows:
\begin{center}
 \begin{tabular}{ll}
  \multicolumn{2}{@{}c@{}}{%
    $\tyinterp{\stmod{\tau}}
       = \{ \mu : \tyinterp{\tau} \ra \mathbb{R}^+ \vbar
           \mathsf{dsupp}(\mu) \mbox{ discrete} \land\sum_{x\in\tau} \mu~x=1 \}$}\\[.5em]
  $\tyinterp{\stmodc{\tau}} = \dotted{\tyinterp{\tau}}$ &
  $\tyinterp{\stfun{\tau}{\sigma}} = \tyinterp{\tau} \cfun \tyinterp{\sigma} ,$
 \end{tabular}
\end{center}
\noindent where the \emph{support} $\mathsf{dsupp}(\mu)$ of a distribution $\mu$
is the set of elements for which $\mu$ takes a non-zero value, and $\cfun$
represents a cpo-continuous function space when the codomain is equipped with a
cpo structure, and the set-theoretical function space otherwise.
\end{definition}
Types $\stmodc{\tau}$ and $\sigma\rightarrow\tau$ where $\tau$ is interpreted as
a cpo are interpreted as cpos. However, types of the form $\stmod{\tau}$ are not
interpreted as cpos, because their interpretation is based on discrete
distributions.\footnote{%
  It would have been possible to interpret them as sub-distributions and to
  define another $\kwletrec$  operator for probabilistic computations, at the
  cost of replacing $\delta+\delta'$ by $\exp(\epsilon')~\delta +
  \exp(\epsilon)~\delta'$ in the typing rule for $\textsf{bind}$.  However, our
  examples do not require this additional generality.}

We can now define the denotational interpretation of expressions.
\begin{definition}
A \emph{valuation} $\theta$ is any finite map from $\vars$ to $\bigcup_\tau
\tyinterp{\tau}$.  A valuation $\theta$ \emph{validates} an environment
$\Gamma$, written $\theta \vDash \Gamma$, if $\forall x \in
\dom(\Gamma) ,\, x\theta \in \tyinterp{x\Gamma}$.
We denote by $\interp{\theta}{\Gamma \vdash e : \tau}$ the
interpretation of $\Gamma \vdash e : \tau$ with respect to $\theta
\vDash \Gamma$.
\end{definition}
The definition of the interpretation is mostly standard; \Cref{fig:interp} gives
the interpretation of the monadic constructions and of the two $\kwletrec$
operators.  As expected, the static semantics is sound w.r.t\, the denotational
one.
\begin{lemma}
 If $\Gamma \vdash e : \tau$ and $\theta \vDash \Gamma$, then
 $\interp{\theta}{e} \in \tyinterp{\tau}$.
\end{lemma}

\begin{myfigure}
 \begin{tabular}{r@{$\;$}l@{\;}l@{\hspace{1cm}}l}
   $\interp{\theta}{\Gamma \vdash \sunitC{e} : \stmodc{\tau}}$
    & = & $\interp{\theta}{\Gamma \vdash e : \tau}$\\[.5em]
   $\interp{\theta}{\Gamma \vdash \sbindC{x}{e_1}{e_2} : \stmodc{\sigma}}$
    & = &\\[.3em]
    & \multicolumn{2}{@{\hspace{-2.5cm}}l}
         {$\left\{\begin{array}{@{\,}l@{\hspace{.3cm}}l}
        \bot & \mbox{if $d = \bot$}\\
        \interp{\mapx{\theta}{x}{d}}{\Gamma \vdash e_2 : \stmodc{\sigma}} & \mbox{otherwise}
      \end{array}\right.$}\\[.7em]
   & \multicolumn{2}{@{\hspace{-0.7cm}}l}
       {where $d = \interp{\theta}{\Gamma \vdash e_1 : \stmodc{\tau}}$}\\[1em]
   $\interp{\theta}{\Gamma \vdash \sunitM{e} : \stmod{\tau}}$
    & = &
      $x \mapsto \left\{\begin{array}{@{}r@{\,}l}
          1 & \mbox{ if $x = \interp{\theta}{e}$}\\
          0 & \mbox{ otherwise}
       \end{array}\right.$\\[1em]
   $\interp{\theta}{\Gamma \vdash \sbindM{x}{e_1}{e_2} : \stmod{\sigma}}$
    & = &\\[.3em]
    & \multicolumn{2}{@{\hspace{-2.5cm}}l}{%
          $d \mapsto \sum_{g \in \tyinterp{\tau}}
                              (\interp{\theta}{e_1}(g)
                       \times \interp{\mapx{\theta}{x}{g}}{e_2}(d))$}\\[.7em]
   $\interp{\theta}{\Gamma \vdash \sletrec{\markvar}{f}{x}{e} : \stfun{\tau}{\sigma}}$
    & = & $\bigcup_{n \in \nats} F^n(\bot_{\tyinterp{\stfun{\tau}{\sigma}}})$\\[.5em]
    & \multicolumn{2}{@{\hspace{-3.3cm}}l}
         {where $F(d) = \interp
           {\mapx{\theta}{f}{d}}
           {\Gamma, f : \stfun{\tau}{\sigma} \vdash
              \slam{x}{e} : \stfun{\tau}{\sigma}}$}\\[.5em]
 \end{tabular}
 \caption{\label{fig:interp}Interpretation of PCF Expressions}
\end{myfigure}

\subsection{Refinement Typing}
\label{subsec:refinement}
The key point of relational typing is its ability to relate a
\emph{pair} of expressions---which we call the \emph{left} and
\emph{right} expressions---via relational assertions that appear as
refinements in types. For instance, the type
\begin{center}
  $\rtprod{x}{\stnat}{\rtref{y}{\stnat}{\l{y} = \l{x} + 1 \wedge \r{y} = \r{x}}}$
\end{center}
represents a pair of integer to integer functions where the left function
adds $1$ to argument, and the right one returns its argument untouched.

In this section, we define the refinement type system of \THESYSTEM in three
steps.  First, we give an interpretation for assertions and refinement types.
Second, we define a subtyping relation that is complete w.r.t.\, this
interpretation. Finally, we define the refinement type system, and prove its
soundness w.r.t.\, a denotational semantics.

We start by defining relational contexts.
\begin{definition}
A \emph{relational environment} $\renv{G}$ is any finite sequence of
relational bindings $(x :: T)$ s.t. a variable is never bound twice
and only variables of $\rvars$ are bound. We use $\emptyset$ to denote
the empty environment. A relational environment defines a finite
mapping from variables to relational types; we write $x\renv{G}$ for
the application of the finite map $\renv{G}$ to $x$.
\end{definition}
We define a \emph{type erasure function} $\rembed{\cdot}$ from relational to
simple types, which maps dependent products to function
spaces, and erases refinements and the indexes of the probabilistic
monad. The definition of $\rembed{\cdot}$ extends recursively to
relational environments: $x\rembed{\renv{G}} = \rembed{x\renv{G}}$ for
any $x \in \dom(\renv{G})$.  We also define the \emph{relational type
erasure of $\renv{G}$}, written $\rrembed{\renv{G}}$, by
$x_\svar\rrembed{\renv{G}} = x\rembed{\renv{G}}$ iff $x\in
\dom(\renv{G})$, where $\svar \in \{ \lvmark, \rvmark \}$. Note that
given a relational binding $(x::T)$, the relational type erasure
$\rrembed{(x::T)}$ gives the environment
$(\l{x}:\rembed{T},\r{x}:\rembed{T})$.

Next, we interpret assertions and refinement types.
\begin{definition}[Relational interpretation of refinement types]
 We say that a valuation $\theta$ \emph{validates a relational environment}
 $\renv{G}$, written $\theta \VDash \renv{G}$, if $\theta \vDash
 \rrembed{\renv{G}}$ and $\forall x \in \dom(\renv{G})$,
 $(\l{x}\theta, \r{x}\theta) \in \rinterp{\theta}{x\renv{G}}$.
 \Cref{fig:finterp,fig:rinterp} define the \emph{relational
 interpretation} $\interp{\theta}{\phi} \in \{ \top, \bot \}$ (resp.
 $\rinterp{\theta}{T} \in \tyinterp{\rembed{T}}^2$) of an
 assertion $\phi$ (resp. of a relational type $T$) w.r.t
 a valuation $\theta \vDash \Gamma$ (resp. $\theta \vDash \rrembed{\renv{G}}$).
\end{definition}

\begin{figure}
  \begin{center}
    $\begin{array}{r@{\;}c@{\;}l}
      \interp{\theta}{\rformc(\phi_1, \ldots, \phi_n)} & = &
        \overline{\rformc}(\interp{\theta}{\phi_1}, \ldots, \interp{\theta}{\phi_n})\\[.5em]
      \interp{\theta}{\rmark{e}_1 = \rmark{e}_2} & = &
        \interp{\theta}{\rmark{e}_1} = \interp{\theta}{\rmark{e}_2}\\[.5em]
      \interp{\theta}{\rmark{e}_1 \le \rmark{e}_2} & = &
        \interp{\theta}{\rmark{e}_1} \le \interp{\theta}{\rmark{e}_2}\\[.5em]
      \interp{\theta}{\fquant{\forall}{x}{\tau}{\phi}} & = &
        \bigwedge_{d \in \tyinterp{\tau}} \interp{\mapx{\theta}{x}{d}}{\phi}\\[.5em]
      \interp{\theta}{\rquant{\forall}{x}{T}{\phi}} & = &
        \bigwedge_{(d_1, d_2) \in \rinterp{\theta}{T}} \interp{\rmapx{\theta}{x}{d_1}{d_2}}{\phi}\\[.5em]
      \interp{\theta}{\fquant{\exists}{x}{\tau}{\phi}} & = &
        \bigvee_{d \in \tyinterp{\tau}} \interp{\mapx{\theta}{x}{d}}{\phi}\\[.5em]
      \interp{\theta}{\rquant{\exists}{x}{T}{\phi}} & = &
        \bigvee_{(d_1, d_2) \in \rinterp{\theta}{T}} \interp{\rmapx{\theta}{x}{d_1}{d_2}}{\phi}\\[.5em]
    \end{array}$
  \end{center}

 \noindent where $\overline{\rformc}$ stands for the $\rformc$-boolean operator.

 \caption{\label{fig:finterp} Relational interpretation of assertions}
\end{figure}

Assertions are interpreted relationally in the expected way where some
care is needed for quantifiers since the interpretation distinguishes
between binders for relational and plain variables.
Relational types are interpreted as sets of pairs of elements of the
interpretation of the erased type. Formally, a pair $(d_1,d_2)$ is in
the relational interpretation of a refinement $\rtref{x}{T}{\phi}$ if
the assertion $\phi$ holds in a relational context where $d_1$ and
$d_2$ are assigned to $\l{x}$ and $\r{x}$, respectively.

The relational interpretation of the dependent product is defined in a
\emph{logical relation} style: it relates function elements $f_1,f_2$
that map related elements $d_1,d_2$ (in $\rinterp{\theta}{T}$) to
related elements (in $\rinterp{\rmapx{\theta}{x}{d_1}{d_2}}{U}$). A
monadic type $\rtmodc{T}$ is relationally interpreted as the set of
pairs in the interpretation of $T$ plus the pair $(\bot,\bot)$. The
polymonadic type $\rtmod{\epsilon}{\delta}{T}$ is interpreted using
a lifting construction $\lift{\epsilon,\delta}{\cdot}$ that turns a relation $\Psi$
on $T_1\!\times\!T_2$ into a relation $\lift{\epsilon,\delta}{\Psi}$
on $\tmod{T_1}\!\times\!\tmod{T_2}$.

\begin{definition}[Lifting of a relation]
  Given $\Psi\subseteq T_1\!\times\!T_2 $, we have
  $\mathrel{\lift{\epsilon,\delta}{\Psi}}~\mu_1~\mu_2$ iff there is a
  distribution $\mu\in\tmod{T_1\times T_2}$ such that
  \begin{enumerate}
    \item $\mu\, (a,b) > 0$ implies $(a,b)\in \Psi$, 
    \item $\pi_1\, \mu\leq\mu_1\land \pi_2\, \mu\leq \mu_2$, and
    \item $\Delta_{\epsilon}(\mu_1, \pi_1\,\mu)\leq \delta ,
      \land\Delta_{\epsilon}(\mu_2, \pi_2\,\mu)\leq \delta$ ,
  \end{enumerate}
where $\pi_1\, \mu = \lambda x. \sum_{y} \mu\, (x,y)$ and $\pi_2\, \mu = \lambda y. \sum_{x} \mu\, (x,y)$.
\end{definition}

\begin{figure}
 \begin{mathpar}
   \inferrule
     { (d_1, d_2) \in \tyinterp{\coty{\tau}}^2 }
     { (d_1, d_2) \in \rinterp{\theta}{\coty{\tau}} }

   \inferrule
     { (d_1, d_2) \in \rinterp{\theta}{T} \\
       \interp{\rmapx{\theta}{x}{d_1}{d_2}}{\phi} }
     { (d_1, d_2) \in \rinterp{\theta}{\rtref{x}{T}{\phi}} }

   \inferrule
     { (f_1, f_2) \in \tyinterp{\rembed{T} \ra \rembed{U}}^2 \\
       \forall (d_1, d_2) \in \rinterp{\theta}{T} .\,
         (f_1(d_1), f_2(d_2)) \in \rinterp{\rmapx{\theta}{x}{d_1}{d_2}}{U} }
     { (f_1, f_2) \in \rinterp{\theta}{\rtprod{x}{T}{U}} }

   \inferrule
     { (d_1, d_2) \in \rinterp{\theta}{T} \cup \{ (\bot, \bot) \} }
     { (d_1, d_2) \in \rinterp{\theta}{\rtmodc{T}} }

   \inferrule
     { \mu_1, \mu_2 \in \stmod{\rembed{T}} \\\\
       \lift{\epsilon, \delta}{\rinterp{\theta}{T}}\ \mu_1\ \mu_2 }
     { (\mu_1, \mu_2) \in \rinterp{\theta}{\rtmod{\epsilon}{\delta}{T}} }
 \end{mathpar}

 \caption{Relational interpretation of types \label{fig:rinterp}}
\end{figure}

\medskip

We next define subtyping between refinement types.
\begin{definition}[Subtyping]
 The \emph{subtyping relation} $\renv{G} \vdash T \tyle U$ is defined by the
 rules of \Cref{fig:subty}.
\end{definition}
A subtyping judgment $\renv{G} \vdash T \tyle U$ relates only
relational types that erase to the same simple type, i.e.
$\rembed{T}=\rembed{U}$. The rules \textsc{Sub-left} and
\textsc{Sub-right}  allow erasing and reinforcing
refinements. The rule \textsc{Sub-M} allows weakening the indices of
the probabilistic polymonad, and the underlying refinements. Other
rules are mostly standard. The definition of subtyping validates the
relational counterpart of the equivalence \textsc{Fun-Sub} discussed
in \Cref{sec:overview}. More generally, it is possible to define a
normalization function that converts any refinement type $T$ into an
equivalent type $\rtref{x}{U}{\phi}$, where $U$ is a simple type,
i.e.\, does not contain any refinement. The existence of the
normalization function immediately entails semantic subtyping.

Finally, we present the \THESYSTEM~typing rules.

\begin{definition}[Relational Typing]
The \emph{refinement typing relation} $\renv{G} \vdash e_1 \sim e_2 :: T$ is
defined in \Cref{fig:relty}. We use $\Gamma \vdash e :: T$ as a shorthand for
$\Gamma \vdash e \sim e :: T$.
\end{definition}

\begin{figure}
 \begin{mathpar}
  \inferrule*[left=Sub-Refl]
   {\renv{G} \vdash T}
   {\renv{G} \vdash T \tyle T}

  \inferrule*[left=Sub-Trans]
   {\renv{G} \vdash T \tyle U \\\\
    \renv{G} \vdash U \tyle V}
   {\renv{G} \vdash T \tyle V}

  \inferrule*[left=Sub-List]
   { \renv{G} \vdash T \tyle U}
   {\renv{G} \vdash \stlist{T} \tyle \stlist{U}}

  \inferrule*[left=Sub-C]
   {\renv{G} \vdash T \tyle U}
   {\renv{G} \vdash \rtmodc{T} \tyle \rtmodc{U}}

  \inferrule*[left=Sub-Left]
   {\renv{G} \vdash \rtref{x}{T}{\phi}}
   {\renv{G} \vdash \rtref{x}{T}{\phi} \tyle T}

  \inferrule*[left=Sub-Right]
   {\renv{G} \vdash T \tyle U \\
    \rrembed{\renv{G}, x :: U} \vdash \phi \\\\
    \forall \theta .\, \theta \VDash \renv{G},x::T \implies \interp{\theta}{\phi}}
   {\renv{G} \vdash T \tyle \rtref{x}{U}{\phi}}

  \inferrule*[left=Sub-M]
   {\renv{G} \vdash T \tyle U \\
    \rrembed{\renv{G}} \vdash \epsilon_i : \rplusinfty \\
    \rrembed{\renv{G}} \vdash   \delta_i : \rplusinfty \\
    \forall \theta .\, \theta \VDash \renv{G},x::T \implies
    \interp{\theta}{       \epsilon_1 \le \epsilon_2 < \infty
                    \wedge \delta_1   \le \delta_2   < \infty}}
   {\renv{G} \vdash \rtmod{\epsilon_1}{\delta_1}{T} \tyle \rtmod{\epsilon_2}{\delta_2}{U}}

  \inferrule*[left=Sub-Prod]
   {\renv{G} \vdash T_2 \tyle T_1 \\
    \renv{G}, x :: T_2 \vdash U_1 \tyle U_2}
   {\renv{G} \vdash \rtprod{x}{T_1}{U_1} \tyle \rtprod{x}{T_2}{U_2}}
 \end{mathpar}

 \caption{\label{fig:subty} Relational Subtyping}
\end{figure}

\begin{figure*}
 \begin{mathpar}
  \inferrule*[left=Var]
   {x \in \dom(\renv{G})}
   {\renv{G} \vdash x :: x\renv{G}}

  \inferrule*[left=Abs]
   {\renv{G}, x :: T \vdash e :: U}
   {\renv{G} \vdash \slam{x}{e} :: \rtprod{x}{T}{U}}

  \inferrule*[left=App]
   {\renv{G} \vdash e_1 :: \rtprod{x}{T}{U} \\
    \renv{G} \vdash e_2 :: T}
   {\renv{G} \vdash e_1\ e_2 :: U\{ x \mapsto e_2 \}}


  \inferrule*[left=LetRecSN]
   {\renv{G}, f :: \rtprod{x}{T}{U} \vdash \slam{x}{e} :: \rtprod{x}{T}{U} \\\\
    \renv{G} \vdash \rtprod{x}{T}{U} \\
    \sncond}
   {\renv{G} \vdash \ssnletrec{f}{x}{e} :: \rtprod{x}{T}{U}}

  \inferrule*[left=LetRec]
   {\renv{G} \vdash \rtprod{x}{T}{\stmodc{U}} \\\\
    \renv{G}, f :: \rtprod{x}{T}{\stmodc{U}} \vdash \slam{x}{e} :: \rtprod{x}{T}{\stmodc{U}}}
   {\renv{G} \vdash \sletrec{}{f}{x}{e} :: \rtprod{x}{T}{\stmodc{U}}}

  \inferrule*[left=Case]
   { \renv{G} \vdash T \\
     \renv{G} \vdash e :: \stlist{\coty{\tau}} \\
     \forall \theta .\, \theta \VDash \renv{G} \implies
         \interp{\theta}{
                           (\l{\rembed{e}} = \snil)
           \Leftrightarrow (\r{\rembed{e}} = \snil)} \\\\
     \renv{G} \vdash e_1 :: T \\
     \renv{G},
         x :: \coty{\tau}, y :: \stlist{\coty{\tau}},
         \{ \l{\rembed{e}} = \l{x} :: \l{y} \wedge
            \r{\rembed{e}} = \r{x} :: \r{y} \}
       \vdash e_2 :: T }
    { \renv{G} \vdash \sifseq{e}{x}{y}{e_2}{e_1} :: T }

  \inferrule*[left=UnitC]
   {\renv{G} \vdash e :: T}
   {\renv{G} \vdash \sunitC{e} :: \rtmodc{T}}

  \inferrule*[left=BindC]
   {\renv{G} \vdash e_1 :: \rtmodc{T_1} \\\\
    \renv{G} \vdash \rtmodc{T_2} \\
    \renv{G}, x :: T_1 \vdash e_2 :: \rtmodc{T_2}}
   {\renv{G} \vdash \sbindC{x}{e_1}{e_2} :: \rtmodc{T_2}}

  \inferrule*[left=UnitM]
   {\rrembed{\renv{G}} \vdash \epsilon : \rplusinfty \\
    \rrembed{\renv{G}} \vdash \delta   : \rplusinfty \\
    \renv{G} \vdash e :: T}
   {\renv{G} \vdash \sunitM{e} :: \rtmod{\epsilon}{\delta}{T}}

  \inferrule*[left=BindM]
   {\renv{G} \vdash e_1 :: \rtmod{\epsilon_1}{\delta_1}{T_1} \\\\
    \renv{G} \vdash \rtmod{\epsilon_2}{\delta_2}{T_2} \\
    \renv{G}, x :: T_1 \vdash e_2 :: \rtmod{\epsilon_2}{\delta_2}{T_2}}
   {\renv{G} \vdash \sbindM{x}{e_1}{e_2} :: \rtmod{\epsilon_1 + \epsilon_2}{\delta_1 + \delta_2}{T_2}}

  \inferrule*[left=Sub]
   {\renv{G} \vdash e :: T \\
    \renv{G} \vdash T \tyle U}
   {\renv{G} \vdash e :: U}

   \inferrule*[left=ARedLeft]
   {e_1 \red e'_1 \\
    \renv{G} \vdash e_1 \sim e_2 :: T}
   {\renv{G} \vdash e'_1 \sim e_2 :: T}

  \inferrule*[left=ACase]
   { \renv{G} \vdash T \\
     \rembed{\renv{G}} \vdash e : \stlist{\coty{\tau}} \\
     \rembed{\renv{G}} \vdash e' : \rembed{T} \\\\
     \renv{G}, \{ \l{\rembed{e}} = \snil \} \vdash e_1 \sim e' :: T \\
     \renv{G},
       x :: \coty{\tau}, y :: \stlist{\coty{\tau}},
         \{ \l{\rembed{e}} = \l{x} :: \l{y} \}
       \vdash e_2 \sim e':: T }
    { \renv{G} \vdash \sifseq{e}{x}{y}{e_2}{e_1} \sim e' :: T }
 \end{mathpar}

 \caption{\label{fig:relty} Relational Typing (Selected Rules)}
\end{figure*}

We briefly comment on some of the typing rules. As in relational Hoare
logic~\cite{Benton04}, we distinguish between synchronous and
asynchronous rules; the latter operate on both expressions of the
judgments, whereas the former operate on a single expression and can
relate expressions that have different shapes. Synchronous rules exist
for the two monads, the two $\kwletrec$ and the dependent product.
Note that the rule for application substitutes the argument of the
application into the result type, and does not impose any value
restriction.

The $\kwcase$ construction is an example of a rule with both a synchronous and
an asynchronous version. The synchronous rule requires a synchronicity
condition: the same branch must be taken in the left and right expressions. For
the case of lists, this is ensured by requiring that the matched lists are
either both empty or both non-empty. In contrast, the asynchronous rule does not
require this condition. The reduction rules close typing under reduction, and is
useful to relate expressions that do not have the same shape.

Refinement typing is sound w.r.t.\, its denotational semantics.
\begin{theorem}[Soundness]\label{thm:soundness}
If $\renv{G} \vdash e_1 \sim e_2 :: T$, then for every valuation
$\theta\models\renv{G}$ we have
$(\interp{\theta}{e_1},\interp{\theta}{e_2})\in \rinterp{\theta}{T}$.
\end{theorem}
It follows that \THESYSTEM accurately models differential privacy.
\begin{corollary}[Differential Privacy]
  If $\vdash e :: \rtref{x}{\sigma}{\Phi} \rightarrow
  \rtmod{\epsilon}{\delta}{\rtref{y}{\tau}{\l{y}=\r{y}}}$
  then $\interp{}{e}$ is $(\epsilon,\delta)$-differentially private
  w.r.t.\, adjacency relation $\interp{}{\Phi}$.
\end{corollary}

We have completed a formalization of Theorem 3.1 in the Coq proof
assistant, assuming an axiomatization of probabilities and lifting.

\subsection{Type-checking}
We have implemented a type-checker for \THESYSTEM. The type-checker
generates proof obligations during type-checking; proof obligations
are sent to SMT solvers via Why3. The type-checker uses a ML-like
syntax and includes a few practical extensions like inductive
datatypes, let expressions, as well as the ability to define logical
predicates and core theories for the datatypes.

All the programs presented in \Cref{sec:dp} and \Cref{sec:auctions}, as well as
some additional examples from the DP literature (private histograms, sums, two
level counters and IDC) were automatically type checked by the implementation,
with the only help of top-level type annotations. See \Cref{tab:examples} for a
summary.

Both the type-checker and the Coq formalization are available at
\url{https://github.com/ejgallego/HOARe2/}.
\begin{table}
  \centering
\begin{tabular}{|c|r|r|}
  \hline
  Example & \# Lines & Verif. time \\
  \hline
  \hline
  {\tt histogram    } &  25 & 2.66 s.\\
  {\tt dummysum     } &  31 & 11.95 s.\\
  {\tt noisysum     } &  55 & 3.64 s.\\
  {\tt two-level-a  } &  38 & 2.55 s.\\
  {\tt two-level-b  } &  56 & 3.94 s.\\
  {\tt binary       } &  95 & 18.56 s.\\
  {\tt idc          } &  73 & 27.60 s.\\
  {\tt dualquery    } & 128 & 27.71 s.\\
  {\tt competitive-b} &  81 & 2.80 s.\\
  {\tt competitive  } &  75 & 4.19 s.\\
  {\tt fixedprice   } &  10 & 0.90 s.\\
  {\tt summarization} & 471 & 238.42 s.\\
  \hline
\end{tabular}
\caption{Benchmarks \label{tab:examples}}
\end{table}

\section{Embedding \DFuzz} \label{sec:dfuzz}
\DFuzz~\citep{GaboardiHHNP13} is a linear dependently typed language
that has been used to verify many examples of differential private
algorithms. In this section, we define a type-preserving embedding
from \DFuzz into \THESYSTEM, and recover soundness of \DFuzz from
\Cref{thm:soundness}. The embedding is interesting for several
reasons. First, it shows that \THESYSTEM is sufficiently expressive to
capture all differentially private examples covered by \DFuzz. Second,
it relates two previously disconnected approaches for verifying
differential privacy. Third, it shows how relational refinements
can internalize logical relations.

For compactness, we only consider a terminating fragment of \DFuzz
with probabilities only over real numbers.
Types and expressions are defined in \Cref{fig:dfuzz-terms}; both are
parameterized by indexes, drawn from two distinct languages. The first
one deals with \emph{sensitivities} (interpreted as elements of
$\rplusinfty$) and the second one deals with \emph{sizes} (interpreted as
natural numbers). Typing judgments are of the form $
\tctxo;\cso;\ctxo\dprov \tmo:\tyo $, where $\tctxo$ is an environment
that records the sensitivity and size variables, $\cso$ is a set of
constraints used in pattern matching, and $\ctxo$ is an environment
containing assignments of the form
$x:!_{R}\tau$. \Cref{fig:dfuzz-typing} gives selected typing rules,
where 
environments are combined by algebraic operations. The environment
$R\cdot \Gamma$ is obtained by taking $x:!_{R\cdot R_i}\typetwo$ for every
$x:!_{R_i}\typetwo\in \Gamma$, while environment addition is defined
as:
\begin{equation*}
  \begin{array}{l}
    (x: !_{R_1}\tyo, \ctxo) +    ( x: !_{R_2}\tyo,\ctxw) = x:!_{R_1 +
      R_2}\tyo, (\ctxo + \ctxw) \\
    (x: !_{R}\tyo, \ctxo) +    \ctxw = x:!_{R}\tyo, (\ctxo + \ctxw)
   \qquad\text{if}\, x\notin\dom(\ctxw)
    \\
\ctxo +   (x: !_{R}\tyo, \ctxw) = x:!_{R}\tyo, (\ctxo + \ctxw)
\qquad \text{if}\, x\notin\dom(\ctxo) .
  \end{array}
\end{equation*}
We refer to~\citet{GaboardiHHNP13} for definitions
and further explanation of the typing rules.

\begin{figure}
$$ \begin{array}[t]{lrl@{\hspace{-2cm}}r} \kappa & ::= & \realone \mid
    \natone & \text{(kinds)}\\ S & ::= & i \mid 0 \mid S + 1 &
    \text{(sizes)} \\ R & ::= & i \mid S \mid \R^{\ge 0} \mid \infty
    \mid R + R \mid R \cdot R
      &
      \text{(sensitivities)}\\
      \tyo,\tyw & ::= &
              \R
              \mid \R[R]
              \mid \stlist{\tyo}[S]
              \mid !_{{R}} \tyo \multimap \tyw
              & \text{(types)}
              \\
&\mid& \forall i : \kappa.\, \tyo
\mid \stmod{\R}\\
        \tmo   & ::= & x
                \mid \N
                \mid \R^{\ge 0} \mid \lambda x.\,\tmo
                & \text{(expressions)}
              \\ & \mid &
                \sletrec{}{f}{x}{\tmo} 
                \mid e_1\; e_2
              \mid
                \Lambda i  .\, \tmo
                \mid \tmo[R]
              \\ & \mid &
                \scase{e} {[\snil \Rightarrow e_1 \mid x :: xs_{[i]} \Rightarrow e_2]}
                \\
       \ctxo, \ctxw & ::= & \emptyset   \mid \ctxo, x : !_R\tyo& \text{(environments)}\\
       \tctxo, \tctxw & ::= & \emptyset \mid \tctxo, i : \kappa & \text{(sens. environments)}\\
       \cso, \csw & ::= & \emptyset   \mid S = 0 \mid S = i + 1 \mid \cso,\cso & \text{(constraints)}\\
  \end{array}$$
  \caption{\DFuzz Types and Expressions}
  \label{fig:dfuzz-terms}
\end{figure}

\begin{figure}
  \centering
\begin{mathpar}
  \inferrule
        { }
  { \tctxo;\cso;\Gamma, x : !_1 \typeone \dprov  x : \typeone } \and
  \inferrule
    { \tctxo;\cso;\Gamma, x: !_R \typeone \dprov  e : \typetwo }
  { \tctxo;\cso;\Gamma \dprov  \lambda x .\, e : !_R \typeone \multimap \typetwo }   \and
  \inferrule
  { \tctxo;\cso;\Gamma \dprov  e_1 : !_R \typeone \multimap \typetwo \\ \tctxo;\cso;\Gamma' \dprov  e_2 : \typeone }
       { \tctxo;\cso;\Gamma + R\cdot\Gamma' \dprov e_1\,e_2 : \typetwo }      \and
  \inferrule
  {\tctxo;\cso; \Gamma \dprov  e : \R }
       { \tctxo;\cso;\infty\cdot \Gamma  \dprov \breturn e:
         \stmod{\R} }
\and
  \inferrule
  {\tctxo;\cso; \Gamma, f: !_\infty (!_R \tyo \lin\tyw), x: !_R \tyo  \dprov  e_1 : \tyw }
       { \tctxo;\cso;\infty\cdot \Gamma \dprov \sletrec{}{f}{x}{\tmo} : !_{R_2} \tyo \lin\tyw }     \and
  \inferrule
  {\tctxo,i:\kappa;\cso; \Gamma, f: !_\infty (\forall i:\kappa.\, \tyo)
  \dprov  e_1 : \tyo \and \text{$i$ fresh in $\cso, \Gamma$ }}
       { \tctxo;\cso;\infty\cdot \Gamma \dprov \sletrec{}{f}{i}{\tmo}
         : \forall i:\kappa.\, \tyo }     \and
  \inferrule
  { \tctxo;\cso;\Gamma \dprov  e_1 : \stmod{\typeone} \\
    \tctxo;\cso;\Gamma', x: !_{\infty} \typeone \dprov  e_2 :\stmod{\typetwo} }
       {\tctxo;\cso; \Gamma + \Gamma' \dprov \slet{x}{e_1}{e_2} : \stmod{\typetwo} }     \and
  \inferrule
    { \tctxo, i : \kappa;\cso; \ctxo \dprov e : \tyo \\\\ \text{$i$ fresh in $\cso, \ctxo$}}
  {  \tctxo;\cso; \ctxo \dprov  \Lambda i:\kappa .\, e : \forall i:\kappa.\,\tyo } \and
  \inferrule
    { \tctxo;\cso;\ctxo \dprov e : \forall i : \kappa.\,\tyo \\\\ \tctxo \models S : \kappa}
  { \tctxo;\cso ; \ctxo \dprov e[S] : \tyo[S/i] } \and
  \inferrule
  { \tctxo;\cso ; \ctxw \dprov e : \stlist{\tyo}[S] \\
    \tctxo ;\cso, S = 0    ; \ctxo \dprov e_l : \tyo \\\\
    \tctxo, i : \natone ;\cso, S = i + 1 ; \ctxo, x : !_R \tyo,xs: !_R\stlist{\tyo}[i] \dprov e_r
  : \tyo}
  { \tctxo;\cso; \ctxo + R \cdot \ctxw \dprov \scase{e}{[\snil \Rightarrow e_l \mid
    x::xs_{[i]}  \Rightarrow e_r ]}: \tyo}
\end{mathpar}
\caption{\DFuzz typing rules}
\label{fig:dfuzz-typing}
\end{figure}

In \DFuzz, types $\tyo, \tyw$ are interpreted as metric spaces, with
associated metrics $d_\tyo, d_\tyw$. Then, the \DFuzz type system
enforces metric preservation~\cite{ReedP10,GaboardiHHNP13}: if $e$ is
well typed in context $\Gamma$, for arbitrary closing substitutions
$\theta_1, \theta_2$ for $\Gamma$, the distance between the
interpretations of $\theta_1(e)$ and $\theta_2(e)$ is upper bounded by
the distance between $\theta_1$ and $\theta_2$. As a particular
instance, \DFuzz expressions of type $!_{R}\tyo \lin\tyw$ correspond
to \emph{$R$-sensitive} functions, i.e.\, functions $f$ such that for
every pairs of inputs $v_1$ and $v_2$,
$d_{\tyw}(f\, v_1,f\, v_2)\leq R\cdot d_{\tyo}(v_1,v_2)$.
We will present an embedding which captures metric preservation as a
relational refinement type.

To this end, we first define the multiplication operation on sensitivities in
more detail. We distinguish two sorts: sensitivities $\R_s = \R^{>0} \cup \{ 0_s,
\infty \}$ and distances $\R_d = \R^{>0} \cup \{ 0_d, \bot \}$.  We interpret
sensitivities $R$ in \DFuzz as sensitivities, while metrics (\Cref{fig:dist})
are interpreted as distances.  We write $s$ and $d$ to range over the respective
sorts. To interpret multiplication, we define a associative and commutative
operator $\diamond$ that maps $\R_s \times \R_s \rightarrow \R_s$ and $\R_d
\times \R_s \rightarrow \R_d$.  The non-standard cases are those involving $0_s,
0_d$ and $\infty, \bot$: $$
\begin{array}{ccc}
s \diamond \infty = \infty &\quad&
d \diamond \infty =
  \begin{cases}
    0_d &\quad \text{if }  d = 0_d \\
    \bot &\quad \text{otherwise}
  \end{cases} \\[4mm]
\bot \diamond s = \bot &&
r \diamond r' = r \cdot r' \quad \text{if }  r, r' \in \R^+ .
\end{array}
$$
Note that $\diamond$ with type $\R_d \times \R_d \to \R_d $ is never used.

Then, the expression
$\disttwo_{\typeone}: \eraseDP{\typeone} \times \eraseDP{\typeone}
\rightarrow \R_d$
will capture the distance function for the \DFuzz type $\typeone$ on
our embedding, where $\eraseDP{\cdot}$ is the erasure function from
\DFuzz types to simple types.
For the sake of readability, we define $\disttwo_{\typeone}$ in usual
mathematical style (see \Cref{fig:dist}), where $\mathsf{sz}$ denotes the
length of a list.
\begin{figure}
\begin{align*}
\disttwo_\R(d_1,d_2) &= |d_1-d_2| \\
\disttwo_{\stlist{\typeone}[S]}(d_1,d_2) &= 
    \begin{cases}
      \displaystyle{\sum_{i\leq n}\disttwo_{\typeone} (d_1^i,d_2^i)}
      &\text{if}\ \mathsf{sz}(d_1)=
      \mathsf{sz}(d_2)=S \\
      \bot &\text{otherwise}
    \end{cases} \\
\disttwo_{\R[R]}(d_1,d_2) &=
    \begin{cases}
      0 &\text{if}\ d_1= d_2=R \\
      \bot &\text{otherwise} \\
    \end{cases} \\
\disttwo_{!_R \typeone\lin \typetwo}(d_1,d_2) &=
  \displaystyle{\max_{d_3,d_4\in \eraseDP{\typeone}}}
  \left(%
  \begin{array}[c]{l}
    \disttwo_{\typetwo}d_1\, d_3, d_2\, d_4) \\ - R \diamond \disttwo_{\typeone} (d_3,d_4)
  \end{array}\right) \\
\disttwo_{\stmod{\R}}(\mu_1,\mu_2) &= \max_{x \in \R}
  \left| \ln \left( \frac{\mu_1(x)}{\mu_2(x)} \right) \right| \\
\disttwo_{\forall i:\kappa.\, \tyo}(d_1, d_2) &= \max_{d \in \kappa}
    \disttwo_{\tyo[d/i]}(d_1\, d, d_2\, d) 
\end{align*}
\caption{Metric induced by \DFuzz types}\label{fig:dist}
\end{figure}

\begin{figure}
$$
\begin{array}{r@{\;}c@{\;}l}
(\R[\sensitermone])^* &= &
\rtref{x}{\R{}}{\l{x}=\r{x} = \sensitermone^*} \\
(\stlist{\tyo}[\sizeitermone])^* &=&
\rtref{x}{\stlist{\eraseDP{\tyo}}}{
\mathsf{sz} (\l{x}) = \mathsf{sz} (\r{x}) = \sizeitermone^*}
\\
(!_R \tyo\lin \typetwo)^* & = & {\tyo^*}\rightarrow{\typetwo^*}
\\
(\forall i : \kappa.~\tyo)^* & = & \rtprod{i}{\{ x::\kappa^*\ |\
  \l{x}=\r{x}\}}{\tyo^*} \\
\stmod{\R}^* &= &\rtmod{0}{0}{\R}
\end{array}
$$
\caption{Translation of \DFuzz types}\label{fig:emb}
\end{figure}

The translation $\cdot^*$ is first defined on sensitivities and sizes, then
on types and expressions, and finally on environments. Sensitivities
and sizes are translated directly as expressions of type $\R_s$
and $\mathbb{N}$ respectively. More interestingly, the translation for
types (given in \Cref{fig:emb}) uses refinements and dependent
products to capture size and sensitivity information. The translation
of expressions is straightforward; the only interesting case are:
$$
(\termone[\sizeitermone])^* = \termone^*\, \sizeitermone^*
\qquad (\Lambda \sizeivarone . \termone)^* =\lambda
\sizeivarone:\kappa.\termone^*
$$
\vspace{-5mm}
\begin{multline*}
(\scase{\termone}{[\snil\to \termone_1\mid
  x::xs_{[\sizeivarone]} \to \termone_2]})^*\\
= \scase{\termone^*}{[\snil\to \termone_1^*\mid
  x::xs \to \termone_2^* \{ \mathsf{sz}(xs)/\sizeivarone\} ]} .
\end{multline*}
The translation of environments is defined inductively. We have
$(\emptyset)^* = \emptyset$ for terms, index environments, and
constraints. Moreover, we define:
$$\begin{array}{rcl}
(\Gamma, x : !_R \typeone)^* & = &  \Gamma^*, x :: \typeone^* \\
(\phi, i : \kappa)^* & = & \phi^*, i ::\rtref{x}{\kappa^*}{\l{x}=\r{x}} \\
(\Phi, S_1=S_2) & = & \Phi^*, \rtref{\_}{\mathbb{B}}{S_1^* = S_2^*} .
\end{array}
$$
Note that the translation of size and sensitivity environments
requires the equivalence of the left and right instances of the
relational variables.  We can show soundness of the embedding.
\begin{theorem}[\DFuzz~Embedding]
\label{thm:dfuzz-translation}
  If $\phi;\iconstraintone;\Gamma\dprov e: \typetwo$ then
  \[
    \phi^*,\iconstraintone^*,\Gamma^* \vdash e^* ::\{ y::  \typetwo^*\ |\ \disttwo_{\typetwo}(\l{y},\r{y})\leq
    \disttwo_\Gamma(\l{\Gamma^*},\r{\Gamma^*})\} ,
  \]
where
$
\disttwo_\Gamma(\l{\Gamma^*},\r{\Gamma^*})
=\sum_{x:!_R \tyo\in \Gamma} R\diamond\disttwo_{\tyo} (\l{x}, \r{x})
$.
\end{theorem}
\begin{proof}
By induction on the derivation of $\phi;\iconstraintone;\Gamma\dprov
e: \typetwo$.  The cases for abstraction and application rely on the
properties of semantics subtyping. In particular, using subtyping we
can pass from
\begin{multline*}
\Pi  (x:: \typeone^*).\{ y:: \typetwo^* \ |\
\disttwo_{\typetwo}(\l{y},\r{y})\leq\\
      \disttwo_{\Gamma,x:!_R\sigma}((\l{\Gamma^*},\l{x}::\tyo^*),(\r{\Gamma^*},\r{x}::\tyo^*))\}
\end{multline*}
to
$$
          \{ f::
            \rtprod{x}{
              \typeone^*
            }{
              \typetwo^*
            }
          \ |\
            \disttwo_{!_R\typeone \lin \typetwo} (\l{f}, \r{f}) \leq \disttwo_{\Gamma}(\l{\Gamma^*},\r{\Gamma^*})
        \}
$$
and vice versa.
Higher order relational refinements are crucial here for internalizing the
metric of \DFuzz at the function type that is essentially a logical
relation argument---this equivalence shows that \emph{two functions
are related if they map related inputs to related outputs}.

Similarly, for the probability distribution case, remembering that we
consider only distributions over base types, we use subtyping to
pass from 
$$
      \rtmod{0}{0}{\{ y::\R\mid \distone_{\R}(y_1,y_2)\leq
          \disttwo_\Gamma(\l{\Gamma^*},\r{\Gamma^*})\}}
$$
to 
$$
\{z::\tmod{\R}\mid \mathcal{L}_{0,0}\Big (\distone_{\R}(y_1,y_2)\leq
          \disttwo_\Gamma(\l{\Gamma^*},\r{\Gamma^*})\Big )\; \l{z}\; \r{z} \}
$$
and vice versa.

The other cases are similar. Interestingly, the case of pattern matching does
not require asynchronous reasoning. Indeed, the refinement type of the
translation of the term under match ensures that the two runs will take the same
branch.
\end{proof}
\noindent Hence typable \DFuzz expressions are differentially
private.
\begin{corollary} If $\emptyset;\emptyset;x:!_\epsilon \typetwo \dprov e:
  \stmod{\mathbb{R}}$ then 
$$
  \rtref{x}{\tau^*}{\disttwo(\l{x},\r{x})\leq 1} \vdash e^*::
  \rtmod{\epsilon}{0}{\rtref{y}{\mathbb{R}}{\l{y}=\r{y}}} .$$
\end{corollary}
\begin{proof}
By \Cref{thm:dfuzz-translation},
$$
x:: \tau^* \vdash e^*::
\rtref{y}{\rtmod{0}{0}{\mathbb{R}}}{
  \disttwo_{\stmod{\mathbb{R}}}{}(\l{y},\r{y}) \leq
  \disttwo_{\Gamma}(\l{x},\r{x})}
$$
with $\Gamma \equiv x:!_\epsilon\typetwo$.
Let $U=\rtref{x}{\tau^*}{\disttwo_{\tau}(\l{x},\r{x})\leq 1}$.
By definition of $\disttwo$ and elementary reasoning about
probabilities,
$$x:: U \vdash e^*:: \rtref{y}{\rtmod{0}{0}{\mathbb{R}}}{
\Delta_\epsilon(\l{y},\r{y}) \leq 0} .
$$
Finally, by semantic subtyping:
$$
x:: U \vdash e^*::
\rtmod{\epsilon}{0}{\rtref{y}{\mathbb{R}}{\l{y}=\r{y}}}  .
$$
\end{proof}
Moreover, \Cref{thm:dfuzz-translation} and \Cref{thm:soundness} give a
direct proof of metric preservation for \DFuzz. 
\begin{theorem}[\DFuzz~Metric Preservation~\cite{GaboardiHHNP13}]
  If $\phi;\iconstraintone;\Gamma\dprov e: \tau$ and $\theta\models
  \phi,\iconstraintone$ and $\theta_1,\theta_2 \models \interp{\theta}{\Gamma}$, then
  \[
    \distone_{\interp{\theta}{\tau}}(\interp{\theta_1,\theta}{e^*},\interp{\theta_2,\theta}{e^*})\leq
    \distone_{\interp{\theta}{\Gamma}}(\theta_1,\theta_2) .
  \]
\end{theorem}
Notice that the above theorem uses three valuations. The valuation $\theta$ is
used for index variables which are equal in the two executions. The other two
valuations $\theta_1,\theta_2$ are used to substitute related values in the two
executions.


\section{Differential Privacy} \label{sec:dp}
\Cref{thm:dfuzz-translation} establishes that every differentially
private algorithm that can be modeled in \DFuzz\ is also captured by
\THESYSTEM. In addition, we present a few previously unverified
algorithms demonstrating the features of our system.

In what follows, we will use some notational shorthands. We write
\lstinline{mlet/munit}, \lstinline{clet/cunit} for the bind/return operations of
the probabilistic and partiality monad. We write $\reals^\geq$ for the type
$\rtref{x}{\reals}{\l{x} \geq \r{x}}$. When a relational variable $x$ is assumed
to be equal in both runs ($\l{x} = \r{x}$), we omit the projection and write $x$
for both $\l{x}$ and $\r{x}$.

\subsection{Private Primitives}
We review two differentially private mechanisms that are used in the
next algorithms. The first mechanism is the \emph{Laplace
  mechanism}~\citep{BDMN05}, which releases a private version of a
numeric value (which can differ in the two runs) by adding noise drawn
from the Laplace distribution.

Formally, the $\epsilon$-private Laplace mechanism takes a real number
$x$ as input and returns $x + \nu$, where $\nu$ is random noise drawn from the Laplace
distribution, which has density function
\[
  F(\nu) = \frac{\epsilon}{2} \exp \left( - \epsilon |\nu|\right) .
\]
If $x$ can differ by at most $s$ in adjacent runs, then the Laplace mechanism is
$(\epsilon s, 0)$ differentially private.  We model this as an operator
\lstinline{lap} with the following type:
\[
\rtprod{x}{\reals}{\rtmod{\epsilon ~|\l{x}-\r{x}|}{0}{\rtref{u}{\reals}{\l{u} =
      \r{u}}}} .
\]

When the output range is non-numeric, a typical tool from differential privacy
is the \emph{Exponential mechanism} \citep{MT07}. Let $B$ be the output range, and
suppose there is a \emph{quality score} $Q : B \rightarrow A \rightarrow
\mathbb{R}$. On an input $a$, the exponential mechanism produces $b \in B$
approximately maximizing $Q\ b\ a$. If for every $b \in B$ and adjacent $a_1,
a_2 \in A$ (wrt. an adjacency relation $\Phi$), the quality score $Q$ satisfies the condition
\[
  |Q\ b\ a_1 - Q\ b\ a_2| \leq s ,
\]
then the Exponential mechanism satisfies $(\epsilon s, 0)$-differential privacy.
We model this as an operator $\lstt{expmech}$ of type
$$\rtref{a}{A}{\l{a}\ \Phi\ \r{a}} \rightarrow \mathfrak{S} \rightarrow
\rtmod{\epsilon s}{0}{\rtref{b}{B}{\l{b}=\r{b}}}$$
with the type $\mathfrak{S}$ of score functions defined as
\begin{equation*}
  \begin{array}[b]{l@{\hspace{0.16em}}l}
    \rtref{b}{B}{\!\l{b}\!=\!\r{b}}&
    \rightarrow\!\rtref{a}{A}{\!\l{a}\ \Phi\ \r{a}}\\
    &\rightarrow\!\rtref{r}{\reals}{\!|\l{r} - \r{r}|\!\leq\!s}
\end{array}
.
\end{equation*}


\subsection{Dual Query Release}
We first focus on the problem of privately answering a large set of queries. The
Laplace mechanism is a simple solution, but it's known that this will add noise
to each query proportional to $\sqrt{k}$ for $k$ queries under
$(\epsilon,\delta)$-privacy. When $k$ is large, the large noise will make the
released answers completely useless.
Fortunately, there is a line of algorithms where noise is added in a carefully
correlated manner, guaranteeing privacy while adding noise proportional only to
$\log k$. We have verified the privacy of one such algorithm, called
$\mathsf{DualQuery}$~\citep{GaboardiAHRW14}. The algorithm is parameterized by a
natural number \lstinline{s} and a set \lstinline{qs} of queries to answer
accurately. The input is the number of rounds \lstinline{t} and database
\lstinline{db}, and the output is a private synthetic database that is accurate
for the given queries. The code of the algorithm is given below:
\begin{lstlisting}
let rec dualquery t db = match t with
 | 0       -> munit []
 | 1 + t'  ->
   mlet curdb     = dualquery t' db        in
   let  quality   = build_quality t' curdb in
   mlet e         = expmech db quality     in
   mlet new_qry   = sampleN s e            in
   let  newrecord = opt new_qry            in
   munit (newrecord :: curdb)
\end{lstlisting}
We encode the database as a list of natural numbers; adjacent databases are
lists of the same length whose distance w.r.t.\, $\disttwo_{\stlist{\stnat}}$
is smaller than 1. Here we consider $\disttwo_{\stlist{\stnat}}$ to be defined
similarly to the distance $\disttwo_{\stlist{\typeone}[n]}$ for list of size $n$
defined in \Cref{fig:dist} but where the $n$ is provided implicitly by the
length of the lists.  We represent the output of the mechanism as a list of
selected records, each encoded as a natural number.

The algorithm performs \lstinline{t} steps, producing one record of
the synthetic database in every round. For each round, we first build
a \emph{quality score} \lstinline{quality}---a function from queries
to real numbers---based on the previously produced records, using the
auxiliary function \lstinline{build_quality}. If we think of the
current records as forming an approximate database, the quality score
measures how poorly the approximation performs on each query.  We then
sample \lstinline{s} queries using the exponential mechanism with this
quality score; queries with higher error are more likely to be
selected. These queries are fed into an optimization function
\lstinline{opt}, which chooses the next record to add to the
approximate database.

The only private operation is the exponential mechanism.  The quality
score we generate at each round \lstinline{i} has sensitivity
\lstinline{i}, and so a draw form the exponential mechanism is
$\lstt{i} \epsilon$-private.  Since $\lstt{i}$ is upper bounded by \lstinline{t} and
there are \lstinline{s} samples per round, the privacy cost per round
is bounded by $\lstt{s} \cdot \lstt{t} \cdot \epsilon$. With
\lstinline{t} rounds in total, the whole algorithm is
$\lstt{s}\cdot \lstt{t}^2\cdot \epsilon$-private.
This guarantee is reflected in the type of
\lstinline{dualquery}:
\[
\begin{array}{ll}
            & \rtref{\lstt{t}}{\N}{\l{\lstt{t}}=\r{\lstt{t}}} \\
\rightarrow &\rtref{\lstt{db}}{\stlist{\stnat}}{\disttwo_{\stlist{\stnat}}
      (\l{\lstt{db}}, \r{\lstt{db}}) \leq 1} \\
\rightarrow & \rtmod{\lstts{s} \cdot \lstts{t}^2 \cdot\epsilon}{0}{
\rtref{\lstt{l}}{\stlist{\N}}{\l{\lstt{l}}=\r{\lstt{l}}}} .
\end{array}
\]
The type states that for two runs with adjacent databases, \lstinline{dualquery}
will return synthetic databases that are $\lstt{s} \cdot \lstt{t}^2 \cdot
\epsilon$ apart, where $\lstt{t}$ is the number of iterations and $\lstt{s}$ is
the number of samples used.


\subsection{Private Counters and the Partiality Monad}
Our second example is a \emph{private counter}. The program takes in a list of
real numbers, and releases a list of running counts.
%
This algorithm is also known as the \emph{binary mechanism} due to
\citet{chan-counter} and has not been verified before; previous
verification work focused on the \emph{two-level} counter from the
same paper.

Suppose the input stream has length $T = 2^n$. The binary
mechanism will return a list of noisy sums, reusing noise to reduce
the improve the accuracy of the sums. The algorithm
proceeds via branching recursion. In the base case, we add Laplace noise to the
single element of the input stream and return. In the recursive case,
we split the input stream into a first and a second half and perform
the recursive call on each half;  we then return the noised sum of the whole
stream together with the result of the recursive calls.
Each output list contains the sums for one ``level'' of the tree; the
first list contains a single sum of length $2^n$, the next contains the two sums
of length $2^{n - 1}$, and so on.

The algorithm terminates, but the simple guard condition implemented
in our tool does not capture termination.%
\footnote{It is of course possible to prove termination using known techniques,
  but we want to demonstrate the partiality monad.}
Hence its formalization is based on the partiality monad and its associated
\lstinline{cunit} and \lstinline{clet} operations:
\begin{lstlisting}
let rec binary n ls = match l with
| []        -> cunit (munit [])
| x :: xs   -> match xs with
  | []      -> cunit (mlet sum = lap x in
                     munit ([sum] :: []))
  | y :: ys ->
    let (left, right) = split l in
    clet cleftN       = binary (n - 1) left  in
    clet crightN      = binary (n - 1) right in
    cunit(mlet leftN  = cleftN               in
          mlet rightN = crightN              in
          mlet sum    = lap (sum l)          in
          munit ([sum] :: (leftN ++ rightN)))
\end{lstlisting}
The algorithm \lstinline{binary} takes as input a natural number
\lstinline{n} and a list \lstinline{ls} of reals with length
$2^{\lstts{n}}$ and returns a list of lists of reals. Formally,
\lstinline{binary} has type
\[
\begin{array}{@{\hspace{0cm}}l@{\hspace{0.1cm}}l}
 & \rtref{\lstt{n}}{\N}{\l{\lstt{n}}=\r{\lstt{n}}} \\
 \rightarrow & \rtref{\lstt{l}}{\stlist{\R}}{
   \lstt{sz}(\l{\lstt{l}}) = \mathsf{sz}(\r{\lstt{l}}) = 2^{\lstts{n}}
   \land \disttwo_{\stlist{\reals}}(\l{\lstt{l}},\r{\lstt{l}}) \leq \lstt{k}} \\
 \rightarrow &
 \rtmodc{\rtmod{\epsilon \cdot \lstts{k} \cdot (\lstts{n}+1)}{0}{\rtref{\lstt{l}}{\stlist{(\stlist{\reals})}}{\l{\lstt{l}}=\r{\lstt{l}}}}} ,
\end{array}
\]
where we write $\lstt{n}$ for readability (since it is assumed equal in both
runs) and where we  use $\disttwo_{\stlist{\reals}}$ to the distance of lists at
the type $\stlist{\reals}$ (defined analogously to
$\disttwo_{\stlist{\stnat}}$).

\section{Auctions and Algorithmic Game Theory} \label{sec:auctions}



We now study the verification of mechanisms with incentive
properties. We start by describing the truthfulness property for
deterministic mechanisms, then we proceed to the randomized case. The
closing examples illustrates the problem of computing an approximate Nash
equilibrium using differential privacy.

\subsection{Truthful auctions}
In the \emph{digital goods} setting, there is an infinite supply of
identical goods to be sold in auction. For instance, when selling music
downloads, goods can be reproduced for free. We assume every agent (or
\emph{bidder}) $i$ has a secret value $v_i$, which is the price she values the
item, and submits a single bid $b_i$ to the mechanism. Once all bids have been
submitted, the mechanism selects a set of winning bidders and prices $p_i$ for
each winner.  Bidders aim to maximize their utility, which is $0$ if they do not
win and $v_i - p_i$ if they win and at price $p_i$.

We want our mechanism to be \emph{truthful}: given fixed bids of the
other agents $b_{-i}$, the utility of agent $i$ is maximized when she
bids her true valuation $b_i = v_i$. This feature makes bidding easy
for bidders, and provides the algorithm designer with some assurance
that she will see the correct inputs (the true values of the players).
From a verification point of view, truthfulness is a \emph{relational}
property of programs: if the mechanism maps bids to outcomes and all
but one of the bidders bids the same in both runs, then the remaining
bidder should have higher utility when bidding truthfully than when
bidding non-truthfully.

We start with the \emph{fixed-price auction}, a very simple mechanism for this
setting. First, we pick a price $p$ (the \emph{reserve price}). Bidders then
submit their bids, and we select all bidders who bid above $p$ to be winners.
Each winning bidder is charged price $p$.

Informally, this process is truthful: a bidder's price does not depend on her
own bid, so lowering her bid will never lower the price---it can only cause her
to lose the item at a price that she would have wanted to pay.  Similarly,
increasing her bid above her value is never beneficial: If her truthful bid is
winning, raising her bid does not change the outcome (she still wins, and pays
the same price). If her truthful bid is losing, raising her bid can only cause
her to win the item at a price that is higher than her value.

To model this with code, we will model a single bidder's utility when she
deviates.\footnote{%
  While it is possible to code the full auction that calculates all the winners
  and charges all the prices, we verify the core guarantee of truthfulness,
  which deals with a \emph{single} bidder's utility function.}
Note that each bidder is treated independently---her utility depends solely on
her value, her bid, and the reserve price, and not on what any of the other
bidders do. So, we can model this auction with the following function, which
calculates a single bidder's utility:
\begin{lstlisting}[mathescape]
let fixedprice b p = if b > p then v - p else 0
\end{lstlisting}
For clarity, we treat \lstinline{v} as a parameter declared in context
with refinement type $\rtref{\lstt{v}}{\reals}{\l{\lstt{v}} =
  \r{\lstt{v}}}$. Truthfulness of this auction follows from the type of
\lstinline{fixedprice}: 
\[
\rtref{\lstt{b}}{\reals}{\l{\lstt{b}} = \lstt{v}} \rightarrow 
\rtref{\lstt{p}}{\reals}{\l{\lstt{p}} = \r{\lstt{p}}}
\rightarrow
\rtref{\lstt{u}}{\reals}{\l{\lstt{u}} \geq \r{\lstt{u}}} .
\]
The relational variable \lstinline{b} is required to be equal to
\lstinline{v} in the first run, and arbitrary on the second run. Then,
the final utility \lstinline{u} cannot be higher on the second run,
demonstrating truthfulness.

This example also demonstrates a boolean version of the asynchronous
typing rule \textsc{ACase} from \Cref{fig:relty}. Since the bid
\lstinline{b} is arbitrary in the second run, the two runs may take
different branches. Indeed, these are the most interesting cases of
the reasoning: If the same branch is taken in both runs, then the
utility is the same in both runs (since the price is the same in both
runs). When different branches are taken, we verify that truthfulness
holds even when deviating from truthful bidding changes the outcome of
the auction.

\subsection{Universal Truthfulness and Randomized Mechanisms}

While the fixed-price auction is very simple, it has poor revenue
properties since the price is set independently of the bids. Setting
it too high will lead to very few goods sold (and hence low revenue),
and setting it too low may sell many goods, but at a price that is
substantially less than bidders would have been willing to pay (again,
low revenue). However, picking the price as a function of
the bids can destroy the auction's truthfulness property. What to do?

It turns out that randomization is a useful way around this problem. In the
random sampling auction due to \citet{GHKSW06}, the bidders are randomly split
into two groups $\lstt{g}_1$ and $\lstt{g}_2$. The fixed-price maximizing
revenue is computed for each group, and then a fixed-price auction is run in
each group---but using the price computed from the \emph{other} group.
Truthfulness holds since the price charged to any bidder remains independent of
\emph{her own} bid.

Since the mechanism is randomized, we want to verify
\emph{truthfulness in expectation}: an individual's expected utility
will never increase if she deviates from bidding her true value. In
fact, the random sampling auction satisfies a stronger property, known as
\emph{universal truthfulness}: a bidder will never be able to gain by
deviating from truthful bidding, even knowing the random coins of the
mechanism.

To model the random sampling auction, we will treat as parameters the value
$\lstt{v}: \reals$ of the single deviating bidder $*$ and the bids of the other
bidders $\lstt{bs} : \stlist{\reals}$; these are again assumed to the same on
both runs.  We define a (deterministic) utility function that takes the bid
(\lstinline{b}) for the deviating bidder $*$, a coin (\lstinline{mygrp})
indicating the group of $*$, and a list of coins (\lstinline{othergrp})
indicating the groups of the other bidders.  Then, utility for $*$ is computed
using the fixed-price auction with reserve price from the bids in the other
group; the optimal reserve price function is
denoted by \lstinline{optfixed}:
\begin{lstlisting}[mathescape]
let utility b (mygrp, othergrp) =
  let ($\lstt{g}_1$, $\lstt{g}_2$) =
    split (mygrp :: othergrp) (b :: bs) in
  if mygroup
    then fixedprice b (optfixed $\lstt{g}_2$)
    else fixedprice b (optfixed $\lstt{g}_1$)
\end{lstlisting}
Universal truthfulness can be seen from the type of \lstinline{utility}:
\[
\begin{array}{ll}
&\rtref{\lstt{b}}{\reals}{\l{\lstt{b}} = \l{\lstt{v}}} \rightarrow
\rtref{\lstt{c}}{\stprod{\stbool}{\stlist{\stbool}}}{\l{\lstt{c}} =
  \r{\lstt{c}}} \\
&\rightarrow \rtref{\lstt{u}}{\reals}{\l{\lstt{u}} \geq \r{\lstt{u}}}   
\end{array}
\]
The type shows that for any realization of the randomness, the utility
is maximized by truthful reporting.

The main auction takes in the real-valued bid \lstinline{b} of $*$, draws
the booleans indicating the groups, and uses the expectation operation to
compute the expected utility of $*$ on this distribution:

\begin{lstlisting}[mathescape]
let auction b =
  mlet me     = flip               in
  mlet others = repeat N flip      in
  let  coins  = munit (me, others) in
  $\Expect$ coins (utility b)
\end{lstlisting}
Above, \lstinline{flip} returns a uniformly random boolean, and has
type
\[
  \rtmod{0}{0}{\rtref{\lstt{c}}{\stbool}{\l{\lstt{c}} = \r{\lstt{c}}}}.
\]
The \lstinline{repeat} function is used to generate a list of $N$ random
booleans (where $N+1$ is the total number of bidders) that are then used to
split the other bidders into two groups.

Truthfulness for the random sampling auction is reflected by the type for
\lstinline{auction}, which computes the expected utility of $*$:
\[
  \rtref{\lstt{b}}{\reals}{\l{\lstt{b}} = \l{\lstt{v}}} \rightarrow
    \rtref{\lstt{u}}{\reals}{\l{\lstt{u}} \geq \r{\lstt{u}}} .
\]
To verify the truthfulness of this auction we rely on monotonicity of
expectation, as captured by the refinement type from~\Cref{sec:overview}.

\subsection{Nash Equilibrium via Differential Privacy}

In this section, we move beyond auctions and consider the more general setting
of \emph{games}. A game is played by a collection of $N$ agents indexed by $i$,
each with a set of possible actions $A_i$ (the \emph{action space}). Given a
vector of actions (one for each player) $a = (a_1, \dots, a_N)$, each agent
receives a (possibly randomized)
payoff $P_i (a_1, \dots, a_N)$; agents seek to maximize their (expected) payoff.
For an example, auctions can be considered as games where each agent's action
space is the space of possible bids, and the payoff of each agent is their
utility for the chosen outcome.

So far, we have considered mechanisms where one action (truthfully reporting) is
a \emph{dominant strategy}: a maximum payoff strategy no matter how the
opponents play.  In general games, like rock-paper scissors, dominant strategies
usually do not exist. In this section, we consider a weaker solution concept:
\emph{approximate Nash equilibrium}.

\begin{definition} \label{def:ANE}
  Let $\alpha \geq 0$. An assignment of agents to actions (a \emph{strategy
    profile}) $(a_1, \dots, a_N)$ is an \emph{$\alpha$-approximate Nash
    equilibrium} if no single agent $i$ can gain more than $\alpha$ payoff from
  a unilateral deviation from $a_i$, assuming that all other players are playing
  according to $a$.  That is, for all agents $i$ and actions $a_i'$,
  \[
    \Expect [P_i (a_1, \dots, a_i, \dots a_N)] \geq \Expect[P_i (a_1,
    \dots, a_i', \dots a_N)] - \alpha .
  \]
\end{definition}

We consider an algorithm for computing an approximate Nash equilibrium when
the payoffs are not publicly known. There are at least two difficulties: first,
a general game may have several Nash equilibria; agents may prefer different
equilibria among the many that exist. This may lead agents to misrepresent
their payoff functions to influence which equilibrium is selected, something we
want to prevent.  Second, payoff functions may consist of sensitive information,
and agents may be unwilling to reveal their payoffs to the mechanism if the
output could disclose their private information.

Somewhat remarkably, using differential privacy to solve the second problem also
solves the first: if we can compute a Nash equilibrium under differential
privacy, then the profile of actions when agents truthfully report their payoff
function is an approximate Nash equilibrium. A full discussion is beyond the
scope of this paper; we present and verify the approximate Nash equilibrium
property for a version of a mechanism due to \citet{CKRW14}, which computes
approximate Nash equilibrium of \emph{aggregative games}.


In an aggregative game, payoffs are a function only of an agent's own action and
a \emph{signal}, a non-negative real number bounded by $k$ that depends on the
aggregated actions of all players. That is, the payoff function for bidder $i$
is of the form $P_i(a_i, S(a_1, \dots, a_N))$, where $S$ is a signal function of
type $A_1 \times \cdots \times A_N \rightarrow [0, k]$.
In code, we will write \lstinline{sign} for the signal function and
\lstinline{k}\, for the bound; they are the same in both runs and we consider
them implicit parameters.

\begin{figure}
  \centering
\begin{lstlisting}
let rec mkSums i br* br =
  mlet s  = lap (sign (br* i) (br i))   in
  match i with
    | 0      -> munit [s]
    | i' + 1 ->
      mlet ss = mkSums i' br* br       $\;\,$in
      munit (s :: ss)
\end{lstlisting}
\begin{lstlisting}
let rec search i br* br sums =
  if |i - nth i sums| < T + 1/2
    then i
    else match i with
      | 0      -> 0
      | i' + 1 -> search i' br* br sums
\end{lstlisting}
\begin{lstlisting}
let expay br* dev* br dev =
  $\Expect$ (mlet    sums = mkSums k br* br      in
      let   $s^\bullet$     = search k br* br sums in
      let   a*   = dev* (br* $s^\bullet$)          in
      let   a    = dev  (br  $s^\bullet$)          in
      let   p*   = pay* a* (sign a* a)  in
      munit p*) ($\lambda$x. x)
\end{lstlisting}
\caption{\label{fig:agg-game} Aggregative game mechanism}
\end{figure}

The mechanism we have verified computes an equilibrium of an
aggregative game, where the payoffs are reported by the agents. We
want to show that reporting the true payoff \emph{and} playing the action
suggested by the mechanism is an approximate Nash equilibrium.

To keep the notation light, we think of all players as having the same action
space, and we consider only two players with actions \lstinline{a*} and
\lstinline{a}, respectively. We consider the player $*$ as the possibly
deviating player, while the other agent is a meta-player, representing all of
the other players in the aggregative game (who do not deviate).

The relevant code for the mechanism is in \Cref{fig:agg-game}. The function
\lstinline{expay} computes the expected pay-off for bidder $*$; it takes as
inputs \emph{best response functions} \lstinline{br*} and \lstinline{br} that
map signals to an agent's highest payoff action (this is how agents report their
payoff), and \emph{deviation functions} \lstinline{dev*} and \lstinline{dev}
that map recommended actions to actual actions. Agents also have a \emph{true}
payoff, \lstinline{pay*} and \lstinline{pay}, also considered as parameters
since they are the same in both runs. The mechanism will not use these functions
in the code; they are only referred to by the refinements.
 
The function \lstinline{expay} performs the following steps:
\begin{enumerate}
\item Use \lstinline{mkSums} to compute a noisy list \lstinline{sums} of
  signals, using the Laplace mechanism;
\item compute a signal $\lstt{s}^\bullet$ such that if both agents choose
  their (self-reported) best action $\lstt{br~s}^{\bullet}$ and
  $\lstt{br*}~\lstt{s}^{\bullet}$ for signal $\lstt{s}^{\bullet}$,
  then the true signal based on strategy profiles $\lstt{a}$ and
  $\lstt{a*}$ (defined next) is close to $\lstt{s}^\bullet$;
\item apply the deviation functions $\lstt{dev}$ and $\lstt{dev*}$ to the
  recommended action $\lstt{br~s}^{\bullet}$ and $\lstt{br*}~\lstt{s}^{\bullet}$
  of each player to produce the strategy profile $\lstt{a}$ and $\lstt{a*}$;
\item calculate the true payoff \lstinline{p*} for the deviating agent on the
  strategy profile;
\item compute the expectation of the payoff \lstinline{p*} for the
  deviating agent.
\end{enumerate}

Players have two opportunities to deviate: they could misreport their best
response function, or they could choose a deviation function to play differently
than their recommendation. We want to show that reporting the true best response
function and following the recommendation (i.e., using the identity function for
deviation) is an approximate Nash equilibrium. As before, we perform the
verification by assigning \lstinline{expay} a relational refinement type where
$*$ behaves truthfully in the left execution, while in the right execution $*$
behaves arbitrarily. Assuming that $\l{\lstt{br*}}$ is the true best response
function corresponding to $\lstt{pay*}$, and that $\l{\lstt{dev*}}$ is the
identity function, and that $\lstt{br}$ and $\lstt{dev}$ coincide on both runs
(we thus omit subscripts), we want to prove (according to \Cref{def:ANE}):
\[
 \lstt{expay}~\l{\lstt{br*}}~\l{\lstt{dev}}~\lstt{br}~\lstt{dev} \\
 \quad \geq \quad
\lstt{expay}~\r{\lstt{br*}}~\r{\lstt{dev*}}~\lstt{br}~\lstt{dev}
- \alpha
\]
for some value $\alpha$. We do this by checking that \lstinline{expay} has type:
\[
\begin{array}{l@{\hspace{0.2cm}}l@{\hspace{0.1cm}}l@{\hspace{0.1cm}}l}
  & \{ \lstt{br*}  & :: \R \to A & \mid \forall \lstt{s},\lstt{a}.~ \lstt{pay*}~(\l{\lstt{br*}}\, \lstt{s} )\,  \lstt{s} \geq  \lstt{pay*}\ \lstt{a}\, \lstt{s} \} \\
 \rightarrow & \{ \lstt{dev*} & :: A \to A  & \mid \forall \lstt{x}.~\l{\lstt{dev*}}~\lstt{x}=\lstt{x} \}  \\
 \rightarrow & \{ \lstt{br}   & :: \R \to A & \mid \l{\lstt{br}}=\r{\lstt{br}} \} \\
 \rightarrow & \{ \lstt{dev}  & :: A \to A  & \mid \forall
 \lstt{a}.~\l{\lstt{dev}}\, \lstt{a} =\r{\lstt{dev}}\, \lstt{a} = \lstt{a} \} \\
 \rightarrow & \{ \lstt{u}    & :: \rplusinfty & \mid \l{\lstt{u}} \geq
 \r{\lstt{u}} - \alpha \} .
\end{array}
\]
We briefly comment on the two most interesting steps in the verification.
First, when calculating $\lstt{s}^\bullet$, the algorithm computes
\lstinline{sums} by adding Laplace noise to each induced best response: we want
to ensure that agents have a limited influence on the chosen signal (and hence
have limited incentive to misreport their best response strategy). For the
accuracy guarantee we need to show that this noise is not too large (which is
the case if the signal and payoff functions satisfy certain low-sensitivity
conditions captured with refinement types).  This is modeled by assigning to the
Laplace mechanism \lstinline{lap} a refinement type capturing \emph{accuracy}:
\[
    \rtprod{\lstt{x}}{\reals}{
      \rtmod{\epsilon |\l{\lstts{x}}-\r{\lstts{x}}| }{\beta}{
        \rtref{\lstt{u}}{\reals}{\l{\lstt{u}}=\r{\lstt{u}} \land |\l{\lstt{x}} - \l{\lstt{u}}| < T }
      }} ,
\]
where $T$ is defined as $| \l{\lstt{x}}-\r{\lstt{x}}| + \frac{1}{\epsilon} \log
\frac{2}{\beta}$. Informally, this states that the added noise is
less than $T$ with probability $1 - \beta$.

The second interesting point is taking the expected value. The expression 
payoff \lstinline{p*} is randomized, and has type
\[
  \rtmod{ \epsilon'}{\delta'}{
 \rtref{\lstt{u}}{\rplus}{ \l{\lstt{u}} \geq \r{\lstt{u}} - \alpha'}}
\]
for some concrete values of $\epsilon', \delta'$ and $\alpha'$. The payoff above
is a probability distribution on real numbers, related by the lifted inequality
relation. We wish to take the expected value of these distributions, in order to
relate the \emph{expected} payoff on the two runs.  However, a priori, it is not
clear how the expected values of these distributions are related. Fortunately,
the expected values are related in a rather simple way, as seen in the following
refinement for $\Expect$:
$$\begin{array}{ll}
  & \rtmod{\epsilon'}{\delta'}{ \rtref{x}{\rplus}{\l{x} \geq \r{x} - \alpha'}} \\
 \rightarrow & \rtref{\lstt{f}}{\rplus\rightarrow\rplus}{
\l{\lstt{f}}=\r{\lstt{f}}=\mathsf{id}} \\
 \rightarrow & \rtref{\lstt{u}}{\rplusinfty}{\l{\lstt{u}} 
\geq \r{\lstt{u}} - \alpha''} ,
\end{array}
$$
where $\alpha''$ is an expression computed from $\epsilon'$, $\delta'$, and
$\alpha'$. That is, taking expectation of two distributions related by the
\emph{lifted} inequality relation yields two real numbers that are approximately
related by the \emph{unlifted}, standard inequality relation on real numbers.
Though not obvious, the soundness of this refinement can be derived from the
definition of expectation and lifting.

From this refinement on the expected payoff for $*$ computed by
\lstinline{expay}, we conclude that truthful reporting and following the
recommended action is an approximate Nash equilibrium.

\section{Related Work} \label{sec:relatedwork}
Our work lies at the intersection of differential privacy, mechanism
design, probabilistic programming languages, and verification.  We
briefly comment on the first three areas (which are too enormous to be covered
here) and elaborate on the most relevant work in program verification.

\paragraph*{Differential privacy}
Differential privacy, first proposed by \citet{BDMN05} and
formally defined by \citet{DMNS06}, has been an area of
intensive research in the last decade. We have touched on a handful of
private algorithms, including an algorithm for computing running
sums~\citep{chan-counter} (part of a broader literature on streaming
privacy), answering large classes of queries~\citep{GaboardiAHRW14}
(part of a broader literature on learning-theoretic approaches to data
privacy).  We refer readers interested in a more comprehensive
treatment to the excellent surveys by \citet{Dwork06,dpsurvey}.

\paragraph*{Mechanism Design}
Mechanism design was introduced to the theoretical computer science
community (with a new focus on efficient implementations) by the
seminal work of \citet{NR99}; see \citet{NRTV07} for a
textbook introduction. It is understood that truthfulness guarantees
can be difficult to prove and verify, so there is a literature giving
generic reductions from mechanism design to algorithm design in
limited settings, but it is known that this is not possible in full
generality \citep{BBHM08,DR14,HL10,CIL12}. Differential
privacy was first proposed as a tool in mechanism design by \citet{MT07}, and
has since found many applications; see \citet{PR13} for a survey of this area.

\paragraph*{Probabilistic programs}
There is a long line of work that develops models of probabilistic programs. The
monadic representation of distributions originates from~\citet{Giry82} and was
further developed in a programming language setting by later work
\citep{Ramsey:2002,Park:2005,Kiselyov:2009,Bornholt+14}. The connections with
machine learning have recently triggered a surge of interest in probabilistic
programming languages. We refer the reader to recent introductory
articles~\citep{Goodman:2013,GordonHNR14} for further information.

\paragraph*{Verification of higher-order programs}
The refinement type discipline was introduced by~\citet{FreemanP91}, and further
developed by others~\citep{DaviesP00,XiP99,DunfieldP04}.  Advances in SMT
solvers have allowed practical systems that support refinement types through SMT
back-ends, for instance \fseven~\citep{BBFGM08}, \fstar~\citep{fstar}, and
\liquidH~\citep{liquid}. Our work is mostly related to a recent variant of
\fstar\ called \rfstar~\citep{rfstar}. Like \THESYSTEM, \rfstar supports
relational reasoning of probabilistic computations. However, \rfstar\ lacks
support for approximate relational refinement types and higher-order
refinements, which are both critical for verifying differential privacy and
game-theoretic properties.

Dependent types is another expressive typing discipline that can be
used to verify properties of functional programs---they also form the
basis of proof assistants like Coq. Examples of dependently typed
languages include Cayenne~\citep{Augustsson98}, Epigram~\citep{epigram},
Idris~\citep{Brady13} and Trellys~\citep{CasinghinoSW14}. Like our system,
Trellys distinguishes between terminating and non-terminating
expressions to ensure logical consistency.

Other prominent approaches for verifying functional programs include
dynamic checking~\citep{FindlerF02} (and its combination with static
type-checking~\citep{Flanagan06,DBLP:conf/esop/WadlerF09,DBLP:conf/popl/GreenbergPW10}),
model-checking~\citep{OngR11} and translating functional programs into
logic~\citep{Vytiniotis+13}.

\paragraph*{Verification of differential privacy and mechanism design}
There has been significant work on language-based techniques
for verifying differentially privacy. \citet{Pierce:2012} defines three
categories: run-time enforcement
(PINQ~\citep{mcsherry.pinq09}, \textsc{Airavat}~\citep{Roy:2010}),
static enforcement (\Fuzz~\cite{ReedP10},
\DFuzz~\citep{GaboardiHHNP13}), and verification-based
enforcement (\textsf{CertiPriv}~\citep{POPL:BKOZ12}). Our work clearly
falls into the last category. All these works are focused on privacy,
rather than accuracy.
See~\citet{POPL:BKOZ12,Pierce:2012} for a more detailed account of related work.

There has been comparatively little work on language-based techniques for
verifying mechanism design.  \citet{BUCS-TR-2008-026} give an interesting
approach, by presenting a programming language for automatically verifying
simple auction mechanisms. A key component of the language
is a type analysis to determine if an algorithm is \emph{monotone}; if
bidders have a single real number as their value (\emph{single-parameter
  domains}), then truthfulness is equivalent to a monotonicity property (e.g.,
see \citet{mu2008truthful}). Their language can be extended by means of user-defined
primitives that preserve monotonicity. The paper shows the use of the language
for verifying two simple auction examples, but it is unclear how this approach
scales to larger auctions.

Finally, \citet{Fang14} propose the use of program synthesis for verifying
truthfulness of auctions. Their approach reduces the verification of auction to
linear constraints that can be handled by an SMT solver. In this respect,
their approach is similar in spirit to ours. However, the constraints they
consider are linear and moreover their technique applies to imperative programs.
The extension to higher order functions is not obvious.

\section{Future Directions} \label{sec:future}
\THESYSTEM is an expressive system of relational refinement types that captures
differential privacy, game-theoretic properties and other relational properties
of probabilistic computations. An exciting direction for further work is to
formally verify more complex mechanisms whose truthfulness guarantees are less
standard. For example, it would be interesting to verify mechanisms that
are merely \emph{truthful in expectation}, for which universally truthful
analogues do not exist (e.g., \citet{DD09}). Such mechanisms use randomization in
a crucial way for their incentive properties, and in addition to being
interesting challenges for verification, are mechanism for which truthfulness is
non-obvious. We also intend to develop a non-relational version of \THESYSTEM
for reasoning about the accuracy of probabilistic computations.

\paragraph*{Acknowledgements}
We thank Zhiwei Steven Wu for good discussions and in particular for suggesting
the aggregative games example. We also thank the anonymous reviewers for their
close reading and helpful comments. This research was supported by the
European
Community's Seventh
Framework grant \#272487, MINECO grant TIN2012-39391-C04-01 (StrongSoft), Madrid
regional project S2009TIC-1465 (\mbox{PROMETIDOS}), and NSF grant CNS-1065060.

\bibliographystyle{abbrvnat}
\bibliography{header,biblio}


\end{document}

%% file: prelude.tex
\newcommand{\THESYSTEM}{\textsf{HOARe}$^2$\xspace}
\usepackage{xspace}
\usepackage{mathpartir}
\usepackage{stmaryrd}
\usepackage{color}
\usepackage{float}
\usepackage{listings}
\usepackage{xfrac}
\usepackage[usenames,dvipsnames]{xcolor}
\usepackage{bbm}
\usepackage{url}
\usepackage{paralist}
\usepackage{bussproofs}
\usepackage{xifthen}

\usepackage{caption}

\pdfoutput=1

\DeclareCaptionFormat{myformat}{#1#2#3\hrulefill}
\newenvironment{myfigure}{\begin{figure}}{\end{figure}}

\usepackage{amsmath,amsthm,amsfonts,amssymb}
\usepackage{mathabx}

\theoremstyle{plain}
\newtheorem{lemma}{Lemma}[section]
\newtheorem{theorem}{Theorem}[section]

\theoremstyle{definition}
\newtheorem{definition}{Definition}[section]

\theoremstyle{corollary}
\newtheorem{corollary}{Corollary}[section]

\renewcommand{\implies}{\Rightarrow}

\usepackage{listings}

\usepackage[T1]{fontenc}
\usepackage{microtype}

\usepackage[scaled]{beramono}
\newcommand\Small{\fontsize{8.2pt}{8.4pt}\selectfont}
\newcommand*\LSTfont{\Small\ttfamily\SetTracking{encoding=*}{-60}\lsstyle}

\newcommand{\lstt}[1]{\mbox{\LSTfont #1}}
\newcommand\lstiny{\fontsize{6.6pt}{6.8pt}\selectfont}
\newcommand{\lstts}[1]{\mbox{\lstiny\ttfamily\SetTracking{encoding=*}{-60}\lsstyle #1}}
 
\definecolor{DarkGreen}{rgb}{0.1,0.5,0.1}
\definecolor{DarkRed}{rgb}{0.5,0.1,0.1}
\definecolor{DarkBlue}{rgb}{0.1,0.1,0.5}
\usepackage[]{hyperref}
\hypersetup{
    unicode=false,          
    pdftoolbar=true,        
    pdfmenubar=true,        
    pdffitwindow=false,      
    pdftitle={},    
    pdfauthor={}
    pdfsubject={},   
    pdfnewwindow=true,      
    pdfkeywords={keywords}, 
    colorlinks=true,       
    linkcolor=DarkRed,          
    citecolor=DarkGreen,        
    filecolor=DarkRed,      
    urlcolor=DarkBlue,          
}
\usepackage{cleveref}

\crefname{section}{\S}{\S}
\Crefname{section}{\S}{\S}

\newcommand{\ra}{\rightarrow}

\def\reals{\ensuremath{\mathbb{R}}\xspace}
\def\nats{\ensuremath{\mathbb{N}}\xspace}

\newcommand{\mapx}[3]{{#1}_{#2}^{#3}}

\def\vars{\ensuremath{\mathcal{X}}\xspace}

\def\pcf{\ensuremath{\mathbf{PCF}}\xspace}

\def\pcfty{\ensuremath{\mathbf{Ty}}\xspace}
\def\pcfcoty{\ensuremath{\mathbf{CoreTy}}\xspace}

\newcommand{\coty}[1]{\widetilde{#1}}

\def\kwcase{\mathtt{case}}
\def\kwwith{\mathtt{with}}

\def\kwif{\mathtt{if}}
\def\kwthen{\mathtt{then}}
\def\kwelse{\mathtt{else}}
\def\kwlet{\mathtt{let}}
\def\kwletrec{\mathtt{letrec}}
\def\kwin{\mathtt{in}}
\def\kwunit{\mathtt{unit}}
\def\kwbind{\mathtt{bind}}

\def\kwlist{\mathtt{list}}

\newcommand{\snmark}{\downarrow}
\newcommand{\markvar}{\mathfrak{m}}
\newcommand{\sncond}{\mbox{$\mathcal{SN}$-guard}}

\newcommand{\slam}[2]{\lambda {#1} .\, {#2}}
\newcommand{\slet}[3]{\kwlet\ {#1} = {#2}\ \kwin\ {#3}}
\newcommand{\sletrec}[4]{\kwletrec^{#1}\ {#2}\ {#3} = {#4}}
\newcommand{\ssnletrec}[3]{\sletrec{\snmark}{#1}{#2}{#3}}
\newcommand{\sif}[3]{\kwif\ {#1}\ \kwthen\ {#2}\ \kwelse\ {#3}}
\newcommand{\sifseq}[5]{\kwcase\ {#1}\ \kwwith\ [\snil \Rightarrow {#5} \vbar \scons{#2}{#3} \Rightarrow {#4}]}
\newcommand{\scase}[2]{\kwcase\ {#1}\ \kwwith\ {#2}}

\newcommand{\sunitM}[1]{\kwunit_{\mathrm{M}}\ {#1}}
\newcommand{\sbindM}[3]{\kwbind_{\mathrm{M}}\ {#1} = {#2}\ \kwin\ {#3}}
\newcommand{\sunitC}[1]{{#1}_\uparrow}
\newcommand{\sbindC}[3]{\kwlet_\uparrow\ {#1} = {#2}\ \kwin\ {#3}}
\newcommand{\sone}[0]{()}
\newcommand{\strue}[0]{\mathtt{true}}
\newcommand{\sfalse}[0]{\mathtt{false}}

\newcommand{\snil}[0]{\ensuremath{\epsilon}}
\newcommand{\scons}[2]{{#1} \mathrel{::} {#2}}

\newcommand{\stfun}[2]{{#1} \rightarrow {#2}}

\newcommand{\stmod}[1]{\mathfrak{M}[{#1}]}
\newcommand{\stmodc}[1]{\mathfrak{C}[{#1}]}
\newcommand{\stunit}[0]{\bullet}
\newcommand{\stbool}[0]{\mathbb{B}}
\newcommand{\stnat}[0]{\mathbb{N}}

\newcommand{\stxreal}[0]{\overline{\mathbb{R}}}
\newcommand{\stprod}[2]{{#1} \times {#2}}

\newcommand{\stlist}[1]{{#1}\;\kwlist}

\def\lvmark{\triangleleft}
\def\rvmark{\triangleright}

\renewcommand{\l}[1]{#1_\lvmark}
\renewcommand{\r}[1]{#1_\rvmark}
\def\svar{\mathfrak{s}}

\newcommand{\rembed}[1]{|{#1}|}
\newcommand{\rrembed}[1]{\|{#1}\|}
\newcommand{\rmark}[1]{{#1}^{\Join}}

\def\pvars{\ensuremath{\mathcal{X}_{\mathcal{P}}}\xspace}
\def\rvars{\ensuremath{\mathcal{X}_{\mathcal{R}}}\xspace}

\newcommand{\expr}[0]{\ensuremath{\mathcal{E}}}
\newcommand{\rexpr}[0]{\ensuremath{\expr^{\Join}}}

\newcommand{\rtypes}[0]{\ensuremath{\mathcal{T}}\xspace}
\newcommand{\rassert}[0]{\ensuremath{\mathcal{A}}\xspace}

\newcommand{\rtmod}[3]{\mathfrak{M}_{#1,#2}[{#3}]}
\newcommand{\rtmodc}[1]{\mathfrak{C}[{#1}]}
\newcommand{\rtprod}[3]{\Pi ({#1} :: {#2}) .\, {#3}}
\newcommand{\rtref}[3]{\{ {#1} :: {#2} \vbar {#3} \}}

\newcommand{\tref}[3]{\{ {#1} : {#2} \vbar {#3} \}}

\def\rfalse{\bot}
\def\rtrue{\top}

\newcommand{\fquant}[4]{{#1}\ ({#2} : {#3}) .\, {#4}}
\newcommand{\rquant}[4]{{#1}\ ({#2} :: {#3}) .\, {#4}}

\def\quantvar{\mathcal{Q}}
\def\rformc{\mathcal{C}}

\newcommand{\rmapx}[4]{%
{#1}{\left\{\substack{\l{#2} \,\mapsto\, {#3}\\ \r{#2} \,\mapsto\, {#4}}\right\}}}

\newcommand{\renv}[1]{\mathcal{#1}}

\def\red{\mathrel{\rightarrow}}

\newcommand{\tyinterp}[1]{\llbracket {#1} \rrbracket}
\newcommand{\interp}[2]{\llbracket {#2} \rrbracket_{#1}}
\newcommand{\rinterp}[2]{\llparenthesis {#2} \rrparenthesis_{#1}}

\def\tyle{\preceq}

\newcommand{\dotted}[1]{\ensuremath{{#1}_\bot}}
\newcommand{\cfun}{\multimap}

\usepackage{mathrsfs}

\DeclareMathOperator{\dom}{dom}

\newcommand{\vbar}[0]{\mathrel{|}}

\newcommand{\fseven}{\textsc{F}7\xspace}
\newcommand{\fstar}{\textsc{F}${}^\star$\xspace}
\newcommand{\liquidH}{\textsc{LiquidHaskell}\xspace}
\newcommand{\rfstar}{\textsc{RF}${}^\star$\xspace}

\newcommand{\tmod}[1]{\mathfrak{M}[{#1}]}

\newcommand{\Expect}{\mathbf{E}}

\newcommand{\lift}[2]{\mathcal{L}_{#1}(#2)}

\def\dprov{\vdash_{D}}
\def\distone{\mathfrak{D}}
\def\disttwo{\mathcal{D}}
\def\realone{\mathrm{r}}
\def\natone{\mathrm{n}}
\def\R{\mathbb{R}}
\def\N{\mathbb{N}}

\def\typeone{\sigma}
\def\typetwo{\tau}

\def\tyo{\typeone}
\def\tyw{\typetwo}

\def\contextone{\Gamma}
\def\contexttwo{\Delta}

\def\ctxo{\contextone}
\def\ctxw{\contexttwo}

\newcommand{\iconstraintone}{\Phi}
\newcommand{\iconstrainttwo}{\Psi}

\def\cso{\iconstraintone}
\def\csw{\iconstrainttwo}

\newcommand{\icontextone}{\phi}
\newcommand{\icontexttwo}{\psi}

\def\tctxo{\icontextone}
\def\tctxw{\icontexttwo}

\def\tmo{\termone}
\def\termone{e}

\newcommand{\DFuzz}{{{\em DFuzz}}\xspace}
\newcommand{\Fuzz}{{{\em Fuzz}}\xspace}

\def\lin{\multimap}

\def\breturn{\mathop{\mathbf{return}}}

\def\sizeivarone{i}

\def\sensitermone{R}

\def\sizeitermone{S}

\newcommand{\eraseDP}[1]{| #1 |}

\newcommand{\rplusinfty}{\ensuremath{\overline{\R}^+}}

\newcommand{\rplus}{\R^+}

\newcommand{\M}[2][]{%
            \ifthenelse{\isempty{#1}}%
              {\mathsf{M}\, {#2}}
              {\mathsf{M}_{#1}\, {#2}}}